\documentclass[journal]{IEEEtran}

\IEEEoverridecommandlockouts

\usepackage{bm}
\usepackage{epsf}
\usepackage{subfigure}
\usepackage{cite}
\usepackage{graphics}
\usepackage{epsfig} 
\usepackage{amsthm}
\usepackage{mathptmx} 
\usepackage{times} 
\usepackage{amsmath} 
\usepackage{amssymb}  
\usepackage{psfrag}
\usepackage{enumerate}
\usepackage{url}
\usepackage{stfloats}
\usepackage{array}
\usepackage{latexsym}
\usepackage{amssymb} 
\usepackage[usenames]{color}
\usepackage[ruled,vlined]{algorithm2e}
\usepackage{wrapfig}

\usepackage{booktabs}
\usepackage{threeparttable}

\newtheorem{definition}{Definition}
\newtheorem{theorem}{Theorem}
\newtheorem{problem}{Problem}
\newtheorem{proposition}{Proposition}
\newtheorem{lemma}{Lemma}

\newtheorem{remark}{Remark}
\newtheorem{example}{Example}

\newcommand{\lfteqn}{\begin{eqnarray} \begin{array}{lllllll}}
\newcommand{\ndeqn}{\end{array} \nonumber \end{eqnarray}}
\newcommand{\Lfteqn}{\begin{eqnarray} \begin{array}{lllllll}}
\newcommand{\Ndeqn}{\end{array}  \end{eqnarray}}

\begin{document}

\title{\LARGE \bf Minimization of  Sensor Activation in Discrete-Event Systems with Control Delays and Observation Delays
}

\author{Yunfeng~Hou, Ching-Yen~Weng, and Peng~Li
\thanks{Yunfeng Hou (yunfenghou@usst.edu.cn) is with the Institute of Machine Intelligence, University of Shanghai for Science and Technology, Shanghai, 200093, China.
Ching-Yen Weng (weng0025@e.ntu.edu.sg) is with the Robotics Research
Centre, Nanyang Technological University, Singapore.
Peng Li (peng.li@hit.edu.cn) is with the Harbin Institute of Technology (Shenzhen), Shenzhen, 518055, China.}}

\maketitle

\begin{abstract}

In discrete-event systems, to save sensor resources, the agent continuously adjusts sensor activation decisions according to a sensor activation policy based on the changing observations.
However, new challenges arise for sensor activations in networked discrete-event systems, where observation delays and control delays exist between the sensor systems and the agent.
 In this paper, a new framework for activating sensors in networked discrete-event systems is established.
 In this framework, we construct a communication automaton that explicitly expresses the interaction process between the agent and the sensor systems over the observation channel and the control channel.
Based on the communication automaton, we can define dynamic observations of a communicated string.
  To guarantee that a sensor activation policy is physically implementable and insensitive to non-deterministic control delays and observation delays, we further introduce the definition of delay feasibility.
We show that a delay feasible sensor activation policy can be used to dynamically activate sensors even if  control delays and observation delays exist.
A set of algorithms are developed to minimize sensor activations in a transition-based domain while ensuring a given specification condition is satisfied.
A practical example is also provided to show the application of the proposed framework.
Finally, we briefly discuss how to extend the proposed framework to a decentralized observation setting.
\end{abstract}


\section{Introduction}


\IEEEPARstart{S}{ystem} observations are achieved via sensors. 
For the sake of security, battery power, sensor availability, and lifespan, sensor activations are costly. 
It is necessary to turn on and off sensors to efficiently make observations for fulfilling certain tasks, such as fault diagnosis and control.
The problem of sensor activation optimization in discrete-event systems (DESs) has been extensively investigated for more than twenty years.
It was first investigated under static observations \cite{yoo02tac,jiang03tac}, where an optimal set of events whose occurrences must be sensed was calculated.
To reduce sensor costs, the concept of sensor activation was introduced in \cite{cassez08fi,cassez07acsd,cassez07tase}, where sensors for the same events can be turned on and off dynamically after different observed event strings.
To compute a minimal language-based SAP, the authors in \cite{dallal14tac} proposed a game structure called the most permissive observer (MPO), which embeds all the valid SAPs for fault diagnosis. 
The MPO was further investigated in \cite{yin19ac} by considering a general class of properties.
How to minimize sensor activations in a language-based solution domain was also investigated in
\cite{wang19tac}, where the solutions do not need to be limited to a game structure.
The authors in \cite{shu13tase,wang2016online} developed online approaches for the minimization of sensor activations.
Sensor activations for the purpose of control or fault diagnosis were also considered in a decentralized observation setting in \cite{lafortune10tac,wang10ac,yin18tac} where the system is monitored by several local agents, and each of them makes its own observation and sensor activation decision.
Arguing that the data of a DES are usually stored as transitions, the authors in \cite{lafortune10tac,weilin16tac} optimized sensor activation policies by limiting the solution domain to the transitions of the systems such that strings ending up at the same state must be followed by the same sensor activation decisions.

\begin{figure}
	\begin{center}
		\includegraphics[width=7.0cm]{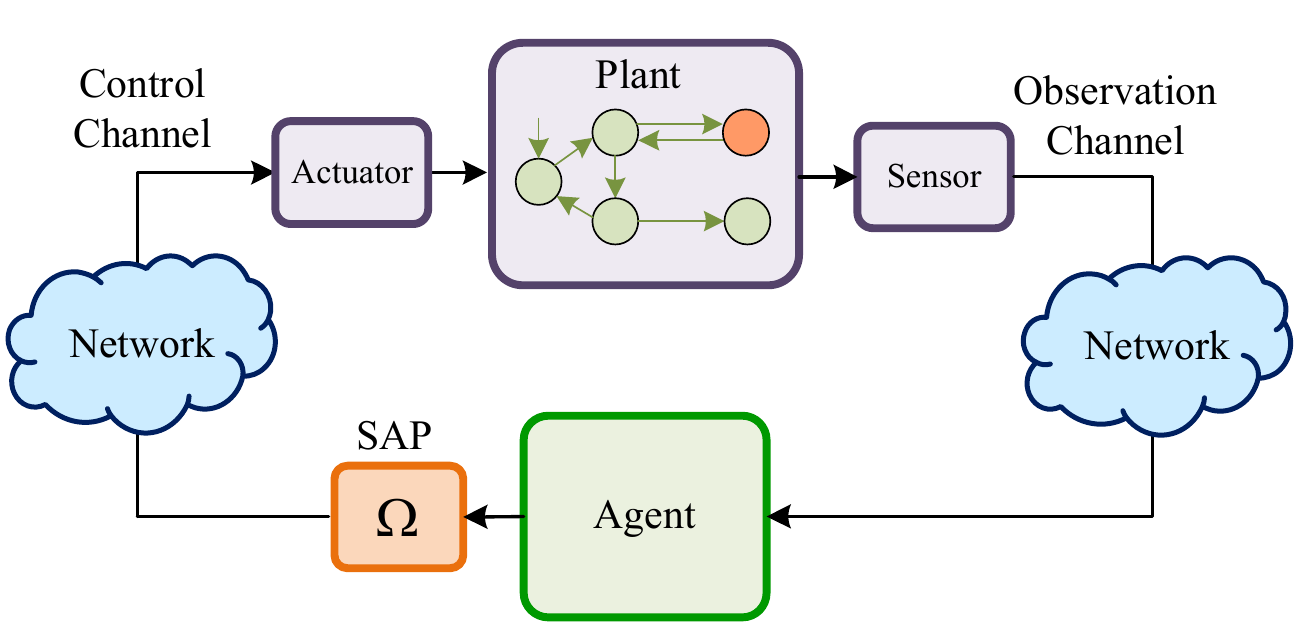}   
		\caption{Sensor activations for networked DESs.} 
		\label{Fig1}
	\end{center}
\end{figure}

\textcolor{blue}{Nowadays, in many modern applications, the system and the agent are connected via networks to communicate sensor readings and  control decisions.
We call such a system as a networked DES, which has drawn considerable attention for supervisory control in the DES community over the past few years \cite{lin2014control,liu21tac, shu17tac1,zhu19cdc,tai22ac,yunfeng23tcns,hou20tac,xu20tcns,zhao17tcns,alves20tac,weilin22ac}. 
Compared with non-networked DESs, the networked structure provides a more flexible and agile way to control a system. 
For example, the agent can be an edge computing node, which can share the computing resources with other edge computing nodes.
However, communicating data over networks is challenging since non-deterministic observation delays and control delays are unavoidable.
On the other hand, even for the non-networked DESs, implementing sensor activation commands could involve preparations after commands are received, which incurs control delays, and detecting whether an event has occurred takes time, which incurs observation delays.
Therefore, a key problem for sensor activations in DESs is how to overcome the  possible delays existing between the agent and the plant.}

In this paper, we propose a new framework to compute a minimal SAP for decision-making, if possible, that is physically implementable and insensitive to possible observation delays and control delays. 
As depicted in Fig. \ref{Fig1}, it is assumed that 1) the occurrences of observable events are detected by sensors and the observations are sent to the agent over the networks; 2) based on the  observations, the corresponding sensor activation decisions made by the agent are sent to the actuator of the plant over the networks.
Due to network characteristics, communication delays exist in both the observation channel and the control channel. 
These delays are non-deterministic and upper bounded but do not change the order of observations and controls.
\textcolor{blue}{As far as we know, this is the first study to address the sensor activation issue of DESs subject to non-deterministic observation delays and control delays.}

To track possible observation delays and control delays, we first model the observation channel and the control channel as an observation channel configuration and a control channel configuration, respectively.
To determine which sensor activation decision is taking effect in the presence of observation delays and control delays, a communication automaton is constructed to dynamically track the plant state, the observation channel configuration, and the control channel configuration.
To guarantee that an SAP is insensitive to observation delays and control delays, it is required that (i) if the SAP intends to activate or deactivate a sensor after the communication of a string, this sensor can be activated or deactivated after this string for all possible control delays and observation delays.
To guarantee that an SAP is physically implementable, it is required that (ii) the sensor activation decision made for every event after two indistinguishable communicated strings must be consistent.
We say that an SAP is delay feasible if both (i) and (ii) are satisfied.
We prove that the union of any two delay feasible SAPs is also delay feasible.
We also prove that the anti-monotonicity property holds for a delay feasible SAP, i.e., the more sensors a delay feasible SAP activates, the fewer string pairs an agent confuses under observation delays and control delays.
We present a set of algorithms to compute a minimal delay feasible transition-based  SAP while ensuring that a given specification condition is satisfied.
Finally, we briefly discuss how to extend the proposed framework to a decentralized observation setting, where several agents are connected to the plant over networks and jointly make sensor activation decisions based on their own sensor readings.

The rest of this paper is organized as follows. 
Section II presents some preliminary concepts. 
Section III discusses how to activate sensors under communication delays.  
Section IV provides algorithms to compute a minimal delay feasible transition-based SAP.
We use a practical example to illustrate the application of the proposed framework in Section V.
Section VI extends the proposed framework to a decentralized system.
Section VII concludes this paper.

\section{Preliminaries}

\subsection{Preliminaries}
We model a DES using a deterministic finite-state automaton $G=(Q,\Sigma,\delta,q_0)$, where $Q$ is the set of states, $\Sigma$ is the set of events, $\delta: Q\times \Sigma \rightarrow Q$ is the transition function, and $q_0$ is the initial state. 
$\delta$ is extended to $Q \times \Sigma^*$ in the usual way \cite{lafortune07book}.  
The language generated by $G$ is denoted by $\mathcal{L}(G)$.
$\varepsilon$ is the empty string.
``$!$'' means ``is defined'', and ``$\neg$!'' means ``is not defined''.
We use $\Sigma_o \subseteq \Sigma$ to denote the set of events being potentially
observable and  $\Sigma_{uo}=\Sigma \setminus \Sigma_o$ to denote the set of events always unobservable.
For each $\sigma \in \Sigma_o$, there is a sensor that can be activated to make it observable (``turning on'' the sensor). 
We denote $TR(G)=\{(q,\sigma) \in Q \times \Sigma:\delta(q,\sigma)!\}$ by the set of all the transitions of $G$.

Given a string $s$, let $|s|$ be the length of a string $s$.
The set of all the prefixes of $s$ is denoted as $\overline{\{s\}}=\{s':(\exists s'')s's''=s\}$.
Given a string $s=\sigma_1\cdots\sigma_k\in \Sigma^*$, we write $s^0=\varepsilon$ and $s^i=\sigma_1\cdots\sigma_i$ for $i=1,\ldots,k$.
Let $s_{-i}$ be the string in $\overline{\{s\}}$ with $|s_{-i}| = \max\{0, |s|-i\}$.
The prefix-closure of a language $L\subseteq \Sigma^*$ is denoted as $\overline{L}$.
$L$ is said to be prefix-closed if $\overline{L}=L$.
$\mathbb{N}$ is the set of natural numbers.
Given a $n \in \mathbb{N}$, let $\Sigma^{\le n}=\{s\in \Sigma^*:|s|\le n\}$ be the set of strings in $\Sigma^*$ whose lengths are less than or equal to $n$.
Given automata $G_1$ and $G_2$, the parallel composition of $G_1$ and $G_2$ is denoted by $G_1||G_2$ \cite{lafortune07book}.
Let $|V|$ be the cardinality of a set $V$.

\textcolor{blue}{The occurrences of transitions of $G$ take time. To describe this, we define the minimum occurring time function $t_{\downarrow}: Q \times \Sigma \rightarrow \mathbb{N}$. Formally, for each transition $(q,\sigma)\in TR(G)$, $t_{\downarrow}(q,\sigma)$ means that event $\sigma$ cannot complete its occurrence at state $q$ unless $t_{\downarrow}(q,\sigma)$ units of time (seconds, for example) have elapsed since it starts to occur.
We extend function $t_{\downarrow}$ to a string in the following manner.
Given a string $s=\sigma_1\cdots \sigma_k \in \Sigma^*$ that is defined at some $q\in Q$, i.e., $\delta(q,s)!$, we write $\delta(q,s^i)=q_i$ for $i=0,1,\ldots,k$.
Then, the minimum occurring time of $s$ at state $q$ is defined as: $t_{\downarrow}(q,s)=t_{\downarrow}(q_0,\sigma_1)+\cdots + t_{\downarrow}(q_{k-1},\sigma_k).$ 
Only a finite number of events can occur in one unit of time.
In this paper, we say that an event occurs if this event completes its occurrence.}

\subsection{Sensor Activation}


In the context of sensor activation, an occurrence of an observable event can no longer be sensed by the agent unless the corresponding sensor is activated when this event happens.
The process of sensor activations in networked DESs is depicted in Fig. \ref{Fig1}, where the event occurrences are communicated to the agent and the sensor activation commands are delivered to the actuator of the plant over a shared network subject to observation delays and control delays.
We make the following assumptions. 1) \textcolor{blue}{The observation delay is upper bounded by $N_o$ units of time, i.e., an event must have been communicated to the agent before no more than $N_o$ units of time since it occurred;} 2) \textcolor{blue}{The control delay is upper bounded by $N_c$ units of time, i.e., a sensor activation command must have been executed by the actuator of the plant before no more than $N_c$ units of time since it was issued;} 3) Both the control channel and the observation channel satisfy the first-in-first-out (FIFO) property, i.e., the sensor activation commands are executed in the same order as they were issued, and the observable event occurrences are received in the same order as they occurred; 4) The actuator always uses the most recently received sensor activation command, and the initial sensor activation command can be executed without any delays since the initial command can be executed beforehand.

The sensor activation command is issued according to an  SAP denoted by $\Omega: \mathcal{L}(G)\rightarrow 2^{\Sigma_o}$.
Specifically, given any $s, s\sigma\in \mathcal{L}(G)$, the sensor for $\sigma\in \Sigma$ after the communication of $s$ should be activated if $\sigma \in \Omega(s)$. 
We call such an SAP $\Omega$ a language-based (or string-based) SAP because the domain of $\Omega$ is the system language of $\mathcal{L}(G)$.
Note that we distinguish ``the occurrence of an event'' and ``the communication of an event'' in this paper.
This is because when observation delays exist, an event may be delayed at the observation channel and not be communicated to the agent immediately after its occurrence. 
We recall the conventional definition of information mapping from \cite{lafortune10tac,wang10ac,dallal14tac,yin18tac,yin19ac} as follows.
Given an SAP $\Omega$, the corresponding information mapping $\theta^{\Omega}:\mathcal{L}({G})\to \Sigma_o^*$ is iteratively defined as: $\theta^{\Omega}(\varepsilon)=\varepsilon$, and for all $s,s\sigma \in \mathcal{L}({G})$, $\theta^{\Omega}(s\sigma)=\theta^{\Omega}(s)\sigma$, if $\sigma \in \Omega(s)$, and $\theta^{\Omega}(s\sigma)=\theta^{\Omega}(s)$, if $\sigma \notin \Omega(s)$.

\textcolor{blue}{To determine whether a sensor for an event is activated, we need to know the sensor activation command taking effect when this event occurs.
If there exist no delays, the sensor activation command taking effect can be uniquely determined based on the current observation.
However, if there are delays, the sensor activation commands taking effect between two successive event occurrences are non-deterministic.
A sensor activation command may take effect after a string for some delays, but may not take effect for the other delays.
It is not an easy task to activate or deactivate a sensor in networked DESs because of the non-deterministic observation delays and control delays. 
We use  Example \ref{Exa1} to illustrate this.}

\begin{figure}
	\begin{center}
		\includegraphics[width=7.4cm]{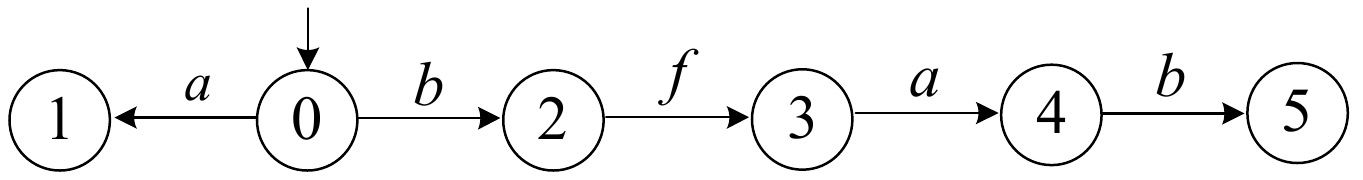}   
		\caption{Automaton $G$ in Example \ref{Exa1}.} 
		\label{Fig2}
	\end{center}
\end{figure}

\begin{example}\label{Exa1}
\textcolor{blue}{Let us consider the system $G$ depicted in Fig. \ref{Fig2}.
Let $\Sigma_o=\{a,b\}$ and $\Sigma_{uo}=\{f\}$.
Let $t_{\downarrow}(2,f)=1$ and $t_{\downarrow}(q,\sigma)=2$ for all $(q,\sigma)\in TR(G) \setminus \{(2,f)\}$.
We want to infer the occurrence of $f$ when $bfa$ is communicated.
To distinguish $bfa$ from $a$, we must activate the sensor for $b$ when the system is in state $0$.
In addition, to distinguish $bfa$ from $b$, we must active the sensor for $a$ when the system is in state $3$.}

\textcolor{blue}{If there are no observation and control delays, we can first activate the sensor for $b$, and then activate the sensor for $a$ after observing $b$.
It can be inferred that event $f$ must have occurred when we see $ba$.
We write the SAP as: $\Omega(\varepsilon)=\{b\}$, $\Omega(b)=\Omega(bf)=\{a\}$, and  $\Omega(bfa)=\Omega(bfab)=\emptyset$.}

\textcolor{blue}{However, the SAP $\Omega$ may fail to work under observation delays and control delays.
For example, let $N_o=1$ and $N_c=2$.
Since $t_{\downarrow}(2,f)=1 \le N_o$ and $t_{\downarrow}(3,a)=2\le N_c$,  let us consider the case that 1) the occurrence of $b$ does not be communicated to the agent until $f$ occurs after $b$; 2) the sensor activation decision made after the communication of $b$ does not be executed until $a$ occurs after $bf$.
If so, the sensor activation command taking effect when event $a$ starts to occur at state 3 is $\Omega(\varepsilon)$.
Since $a \notin \Omega(\varepsilon)$,  we cannot detect the occurrence of $a$ after $bf$.
Moreover, since $f\in\Sigma_{uo}$, we cannot distinguish $bfa$ and $b$ and thus does not know if $f$ has occurred when $bfa$ is communicated.
Furthermore, since $a \in \Omega(b)$ and $a \notin \Omega(bfa)$, we do not know if the sensor for $a$ should be activated when we observe $b$.
Thus, it is necessary to develop a new framework to determine a SAP that is physically feasible and insensitive to observation delays and control delays.}

\end{example}

\section{Sensor Activation for Networked DESs}

In this section, we develop a framework for sensor activations in networked DESs.
In the proposed framework, we explicitly model  the ``dynamics of the observation channel'' and ``dynamics of the control channel''.
We then construct a communication automaton to describe how the agent and the plant interact with each other over the observation channel and the control channel.
Based on the communication model, we can define the observation mapping of an occurred string, even if there exist observation delays and control delays.

\subsection{Modeling Communication Channels}

We first introduce the definition of observation channel configuration as follows.

\begin{definition}\label{DefI}
\textcolor{blue}{(Observation Channel Configuration)} The {observation channel configuration} is defined as a sequence of pairs $\theta_{o}=(\sigma_1,n_1)\cdots (\sigma_k,n_k) \in (\Sigma \times [0,N_o])^*$, where $\sigma_1\cdots\sigma_k$ is a string (in the same order as they were generated) that has occurred but not been communicated due to observation delays, and $n_i$ is the minimum occurring time of the string occurring when $\sigma_i$ is delayed at the observation channel.
\end{definition}

We denote by $\Theta_{o} \subseteq (\Sigma\times [0,N_o])^{\le N}$ the set of all the possible observation channel configurations, where $N$ is the maximum length of an $\theta_o\in \Theta_o$. 
Given an  $\theta_o \in \Theta_o$, if $\theta_o=(\sigma_1,n_1)\cdots(\sigma_k,n_k)\neq\varepsilon$, we let $\mathbf{MAX}^{obs}(\theta_o)=n_1$ be the maximal observation delays, and if $\theta_o=\varepsilon$, we let $\textbf{MAX}^{obs}(\theta_o)=0$.
\begin{remark}\label{Rem1}
\textcolor{blue}{Since the observation delays are upper bounded by $N_o$ units of time,  $N$ is determined by the maximum length of the string whose  minimum occurring time is no larger than $N_o$.
Specifically, let $s \in \Sigma^*$ be a string that is defined at some $q\in Q$ in $G$, i.e., $(\exists q\in Q)\delta(q,s)!$.
Without loss of generality, suppose that $s$ is the longest such string with $t_{\downarrow}(q,s) \le N_o$.
By assumption, if an event $\sigma$ occurs, it will be communicated in the worst case when a string with the length of $|s|$ has occurred after $\sigma$.
Thus, by Definition \ref{DefI}, the maximum length of an observation channel configuration is $|s|+1$, i.e., $N=|s|+1$. 
Since only a finite number of events can occur in $N_o$ units of time, $|s|$ is finite.
Therefore, $N=|s|+1$ is finite.}
\end{remark}

To update $\theta_o$, we introduce the following two operators.

\begin{enumerate}
\item 
When an event $\sigma \in \Sigma$ occurs at state $q$, to update $\theta_o$, we define the operator $\textbf{IN}^{obs}:\Theta_o  \times Q \times \Sigma \rightarrow \Theta_o$ as follows: for all $\theta_{o} \in \Theta_o$ and all $q \in Q$ and all $\sigma \in \Sigma$, 
  \Lfteqn\label{Eq2}
  \textbf{IN}^{obs}(\theta_o,q,\sigma)= \begin{cases}
   \theta_o'  &\text{if}\ \textbf{MAX}^{obs}(\theta_o^+) \le N_o\\ 
\neg !  &\text{if}\ \textbf{MAX}^{obs}(\theta_o^+) > N_o
  \end{cases}
  \Ndeqn  
where if $\theta_o=(\sigma_1,n_1)\cdots(\sigma_k,n_k) \neq \varepsilon$, then $\theta_o'=(\sigma_1,n_1+t_{\downarrow}(q,\sigma))\cdots(\sigma_k,n_k+t_{\downarrow}(q,\sigma))$, and if $\theta_o=\varepsilon$, then $\theta_o'=(\sigma,0)$.

\item  When a new event $\sigma \in \Sigma$ is communicated to the agent, to update $\theta_o$, we define the operator $\textbf{OUT}^{obs}:\Theta_o \times \Sigma \rightarrow \Theta_o$ as follows: for all $\theta_{o} \in \Theta_o$ and all $\sigma \in \Sigma$,
 \Lfteqn\label{Eq3}
  \textbf{OUT}^{obs}(\theta_o,\sigma)= \begin{cases}
   \theta_o'  &\text{if}\ \theta_o=(\sigma_1,n_1)\cdots\\
& (\sigma_k,n_k) \neq \varepsilon \wedge \sigma=\sigma_1\\
  \neg ! & \text{otherwise}
  \end{cases}
  \Ndeqn  
where $\theta'_o=(\sigma_2,n_2)\cdots(\sigma_k,n_k)$.
\end{enumerate}

{By assumption, the minimum occurring time of the string happening when $\sigma_1$ is delayed at the observation channel should never exceed $N_o$.
Thus, as shown in (\ref{Eq2}), $\textbf{IN}^{obs}(\theta_o,q,\sigma)$ is defined iff $\textbf{MAX}^{obs}(\theta_o^+) \le N_o$.
When an event $\sigma \in \Sigma$ occurs, by FIFO, $\textbf{IN}^{obs}(\theta_o,q,\sigma)$ adds $(\sigma,0)$ to the end of $\theta_o$.
Additionally, we add $t_{\downarrow}(q,\sigma)$ to $n_i$ to update the minimum occurring time after the occurrence of $\sigma$.}
On the other hand, if an event $\sigma$ is communicated, by FIFO, $\sigma$ must be the first component of $\theta_o$. 
As shown in (\ref{Eq3}), $\textbf{OUT}^{obs}(\theta_o,\sigma)!$ iff $\theta_{o}=(\sigma_1,n_1) \cdots (\sigma_k,n_k) \neq\varepsilon \wedge \sigma_1=\sigma$.
When $\sigma$ is communicated,  $\textbf{OUT}^{obs}(\theta_o,\sigma)$ removes $(\sigma_1,n_1)$ from the head of $\theta_o$.

\begin{example}\label{Exa2}
{We continue with Example \ref{Exa1}.
Let ${\theta}_o=(b,0)\in \Theta_o$.
Since $t_{\downarrow}(2,f)=1$, when  $f$ occurs after $b$, by (\ref{Eq2}), $\theta'_o=\mathbf{IN}^{obs}({\theta}_o,2,f)=(b,1)(f,0)$.
Since $t_{\downarrow}(3,a)=2$, event $a$ cannot occur after $bf$ unless event $b$ has been communicated to the agent.
When the occurrence of $b$ is communicated, by (\ref{Eq3}),  $\mathbf{OUT}^{obs}({\theta}'_o,b)=(f,0)$.}
\end{example}

To model the dynamics of the control channel, we define the control channel configuration as follows.

\begin{definition}\label{Def2}
\textcolor{blue}{(Control Channel Configuration)} The {control channel configuration} is defined as a sequence of pairs: $\theta_{c}=(\sigma_1,m_1) \cdots (\sigma_k,m_k) \in (\Sigma \times [0,N_c])^*$, where $\sigma_1\cdots \sigma_k \in \Sigma^*$ is a string that has been communicated, but the sensor activation commands made after these communications are not executed, and $m_i$ is the minimum occurring time of the string occurring when the sensor activation decision made after the communication of $\sigma_i$ is delayed at the control channel.
\end{definition}

We denote by $\Theta_{c} \subseteq (\Sigma \times [0,N_c])^{\le M}$  the set of all the possible control channel configurations, where $M\in\mathbb{N}$ is the maximum length of $\theta_c\in\Theta_c$.
Given an $\theta_c \in \Theta_c$, if $\theta_c=(\sigma_1,m_1)\cdots(\sigma_k,m_k)\neq\varepsilon$, let $\mathbf{MAX}^{ctr}(\theta_{c})=m_1$ be the maximal control delays, and if $\theta_c=\varepsilon$, let $\textbf{MAX}^{ctr}(\theta_{c})=0$.

\begin{remark}
\textcolor{blue}{By assumption, an occurred event can be communicated before no more than $N_o$ units of time. When this event is communicated, a sensor activation command is made and can be executed before no more than $N_c$ units of time. 
Thus, a sensor activation command delayed at the control channel can be the one made after an event occurring in the past $N_o+N_c$ units of time.
As discussed in Remark \ref{Rem1}, $M$ is a finite constant and determined by the maximum length of a string whose minimum occurring time is no larger than $N_o+N_c$.}
\end{remark}

To update $\theta_c$, we introduce the following three operators.
\begin{enumerate}
\item 
When a new event $\sigma$ occurs at state $q$, we define $\textbf{PLUS}:\Theta_c \times Q \times \Sigma \rightarrow \Theta_c$ as follows: for all $\theta_c \in \Theta_c$ and all $q \in Q$ and all $\sigma \in \Sigma$,
  \Lfteqn\label{Eq4}
  \textbf{PLUS}(\theta_c,q,\sigma)= \begin{cases}
   \theta_c'  &\text{if}\ \textbf{MAX}^{ctr}(\theta'_{c})\le N_c \\
\neg !  &\text{if}\ \textbf{MAX}^{ctr}(\theta'_c)> N_c
  \end{cases}
  \Ndeqn   
where if $\theta_c=(\sigma_1,m_1)\cdots(\sigma_k,m_k)\neq\varepsilon$,  $\theta_c'=(\sigma_1,m_1+t_{\downarrow}(q,\sigma))\cdots(\sigma_k,m_k+t_{\downarrow}(q,\sigma))$, and if $\theta_c=\varepsilon$,  $\theta_c'=\varepsilon$.

\item 
When an event $\sigma \in \Sigma$ is communicated, we define $\textbf{IN}^{ctr}:\Theta_c \times \Sigma \rightarrow \Theta_c$ as follows: for all $\theta_c \in \Theta_c$ and all $\sigma \in \Sigma$,  
\Lfteqn\label{Eq5} 
\textbf{IN}^{ctr}(\theta_c,\sigma)=\theta_c(\sigma,0).
\Ndeqn 
\item When the sensor activation command made after the communication of $\sigma$ is executed, we define  $\textbf{OUT}^{ctr}:\Theta_c \times \Sigma \rightarrow \Theta_c$ as follows: for all $\theta_c \in \Theta_c$ and all $\sigma \in \Sigma$,
 \Lfteqn\label{Eq6}
  \textbf{OUT}^{ctr}(\theta_c,\sigma)= \begin{cases}
  \theta_c'  &\text{if}\ \theta_c=(\sigma_1,m_1)\cdots\\
&(\sigma_k,m_k)\neq\varepsilon \wedge  \sigma=\sigma_1\\
    \neg ! & \text{otherwise}
  \end{cases}
  \Ndeqn
where $\theta'_c=(\sigma_2,m_2)\cdots(\sigma_k,m_k)$.
\end{enumerate}

When a new event $\sigma$ occurs at state $q$, as shown in (\ref{Eq4}),  $\textbf{PLUS}(\theta_c,q,\sigma)$ adds $t_{\downarrow}(q,\sigma)$ to all the natural numbers in $\theta_c$.
Since control delays are upper bounded by $N_c$, $\textbf{PLUS}(\theta_c,q,\sigma)$ is defined iff $\textbf{MAX}^{ctr}(\theta_{c}')\le N_c$.
\textcolor{blue}{When a new event $\sigma$ (observable or not) is communicated to the agent, a sensor activation command can be issued. 
Then, by Definition \ref{Def2}, we should update the control channel configuration by adding $(\sigma,0)$ to the end of $\theta_c$.
As shown in (\ref{Eq5}), we define $\textbf{IN}^{ctr}(\theta_c,\sigma)$ to achieve this.
Note that an issued sensor activation command can be new only if a new observable event is communicated.}
When a sensor activation command is executed, by FIFO, it must be the one made right after the communication of the first component of $\theta_c$.
Thus, $\textbf{OUT}^{ctr}(\theta_c,\sigma)$ is defined if and only if $\theta_c=(\sigma_1,m_1)\cdots(\sigma_k,m_k)\neq \varepsilon \wedge  \sigma=\sigma_1$.
When the sensor activation command issued after the communication of $\sigma$ is executed, as shown in (\ref{Eq6}), $\textbf{OUT}^{ctr}(\theta_c,\sigma)$ removes the first component from the head of $\theta_c$.

\begin{example}\label{Ex3}
We continue with Example \ref{Exa1}.
Let ${\theta}_c=(b,0)\in \Theta_c$.
Since $t_{\downarrow}(2,f)=1$, when  $f$ occurs after $b$, by (\ref{Eq4}), $\theta'_c=\mathbf{PLUS}(\theta_c,2,f)=(b,1)$.
Then, if  $f$ is communicated, by (\ref{Eq5}), $\theta''_c=\mathbf{IN}^{ctr}(\theta'_c,f)=(b,1)(f,0)$.
Since $t_{\downarrow}(3,a)=2$ and $N_c=2$, by (\ref{Eq4}), event $a$ cannot occur at state 3 unless the command made after the communication of $b$ is executed.
When this command is executed, by (\ref{Eq6}), $\mathbf{OUT}^{ctr}(\theta''_c,b)=(f,0)$.
\end{example}

\subsection{Communication Automaton}

Next, we build a communication automaton to track observation delays and control delays that may occur when the agent communicates with the plant.
We  first introduce two special types of events.

\begin{enumerate}
\item  To keep track of what has been successfully communicated thus far, define bijection $h:\Sigma \rightarrow \Sigma_h$ such that $\Sigma_h=\{h(\sigma):\sigma \in \Sigma\}$. 
\textcolor{blue}{For all $\sigma \in \Sigma$,  we use $h(\sigma)$ to denote that the event occurrence of $\sigma$ has been communicated to the agent.}

\item  To model which sensor activation command is taken,  define bijection $g: \Sigma \rightarrow \Sigma_g$ such that $\Sigma_g=\{g(\sigma):\sigma \in \Sigma\}$. 
\textcolor{blue}{For all $\sigma \in \Sigma$,  we use $g(\sigma)$ to denote that the sensor activation command made immediately after the communication of $\sigma$ has been executed.}
\end{enumerate}
We define the inverse mapping $g^{-1}:\Sigma_g^* \to \Sigma^*$ as,  for all $g(\sigma) \in \Sigma_g$, $g^{-1}(g(\sigma))=\sigma$.
We extend $g$ and $g^{-1}$ to a set of strings in the usual way.
\textcolor{blue}{Since we distinguish ``the occurrence of an event", ``the communication of an event", and ``the execution of a sensor activation command'' in this paper, we assume that $\Sigma$, $\Sigma_h$, and $\Sigma_g$ are mutually disjoint.}

Formally,  we construct the communication automaton $\tilde{G}=(\tilde{Q}, \tilde{\Sigma}, \tilde{\delta}, \tilde{q}_{0})$,  where  $\tilde{Q} \subseteq Q \times \Theta_o \times \Theta_c$ is the state space;
$\tilde{q}_{0}=(q_0,\varepsilon,\varepsilon)$ is the initial state; 
$\tilde{\Sigma} \subseteq \Sigma \cup \Sigma_h \cup \Sigma_g$ is the event set;
The transition function $\tilde{\delta}:\tilde{Q} \times \tilde{\Sigma} \rightarrow \tilde{Q}$ is defined as:

\begin{itemize}
  \item For all $\tilde{q}=(q,  \theta_o, \theta_c) \in \tilde{Q}$ and all $\sigma \in \Sigma$, 

  \Lfteqn\label{Eq7}
  \tilde{\delta} (\tilde{q}, \sigma)= \begin{cases}
  \tilde{q}'  & \text{if}\ \delta(q,\sigma)! \wedge \textbf{IN}^{obs}(\theta_o,q,\sigma)! \\
& \wedge\textbf{PLUS}(\theta_c,q,\sigma)!\\
   \neg ! & \mbox{otherwise,}
  \end{cases}
  \Ndeqn 
where $\tilde{q}'= (\delta(q,\sigma), \textbf{IN}^{obs}(\theta_o,q,\sigma), \textbf{PLUS}(\theta_c,q,\sigma))$;

\item For all $\tilde{q}=(q,\theta_o, \theta_c) \in \tilde{Q}$ and all $h(\sigma) \in \Sigma_h$, 

  \Lfteqn\label{Eq8}
  \tilde{\delta}(\tilde{q}, h(\sigma))= \begin{cases}
   \tilde{q}'  &\text{if}\ \textbf{OUT}^{obs}(\theta_o,\sigma)!\\ 
   \neg ! & \text{otherwise,}
  \end{cases}
  \Ndeqn 
where $\tilde{q}'=(q, \textbf{OUT}^{obs}(\theta_o,\sigma), \textbf{IN}^{ctr}(\theta_c,\sigma))$;

\item For all $\tilde{q}=(q, \theta_o, \theta_c) \in \tilde{Q}$ and all $g(\sigma) \in \Sigma_g$, 

  \Lfteqn\label{Eq9}
  \tilde{\delta}(\tilde{q},  g(\sigma))= \begin{cases}
   \tilde{q}'  &\text{if}\ \textbf{OUT}^{ctr}(\theta_c,\sigma)!\\ 
   \neg ! & \text{otherwise,}
  \end{cases}
  \Ndeqn  
 where $\tilde{q}'=(q, \theta_o, \textbf{OUT}^{ctr}(\theta_c,\sigma))$.
\end{itemize}

We denote each state of $\tilde{G}$ by a triplet $\tilde{q}=(q,\theta_o,\theta_c) \in  Q \times \Theta_o \times \Theta_c$ such that $q \in Q$ tracks the plant state, $\theta_o \in \Theta_o$ tracks the observation channel configuration, and $\theta_c \in \Theta_c$ tracks the control channel configuration. 
For any $(q,\theta_o,\theta_c) \in \tilde{Q}$, by assumption, an event $\sigma \in \Sigma$ can occur at $(q,\theta_o,\theta_c)$ iff $\sigma$ is active at $q$, i.e., $\delta(q,\sigma)!$, and the observation delays and the control delays after the occurrence of $\sigma$ are no larger than $N_o$ and $N_c$, respectively, i.e., $\textbf{IN}^{obs}(\theta_o,q,\sigma)!  \wedge \textbf{PLUS}(\theta_c,q,\sigma)!$. 
Thus,  in (\ref{Eq7}),  $\tilde{\delta} (\tilde{q}, \sigma)!$ iff $\delta(q,\sigma)!\wedge \textbf{IN}^{obs}(\theta_o,q,\sigma)!  \wedge \textbf{PLUS}(\theta_c,q,\sigma)!$.
When $\sigma$ occurs, we set $q \leftarrow \delta(q,\sigma)$.
Meanwhile,  by (\ref{Eq2}) and (\ref{Eq4}), we set $\theta_o \leftarrow \textbf{IN}^{obs}(\theta_o,q,\sigma)$ and $\theta_c \leftarrow \textbf{PLUS}(\theta_c,q,\sigma)$.

In (\ref{Eq8}), for any $(q, \theta_o,\theta_c) \in  \tilde{Q}$, by FIFO, $\sigma$ can be communicated iff the first event of $\theta_o$ is $\sigma$, i.e.,   $\textbf{OUT}^{obs}(\theta_o,\sigma)!$.
When $\sigma$ is communicated, by (\ref{Eq3}), $\theta_o \leftarrow \textbf{OUT}^{obs}(\theta_o,\sigma)$.
Upon the communication of $\sigma$, a sensor activation command is made.
By (\ref{Eq5}),  $\theta_c\leftarrow \textbf{IN}^{ctr}(\theta_c,\sigma)$.
Note that a new sensor activation command can be made only if $\sigma$ is observable, i.e., the sensor for $\sigma$ is activated at the time it occurs.

In (\ref{Eq9}), for any $(q,\theta_o,\theta_c) \in \tilde{Q}$, by FIFO, the sensor activation command made after the communication of $\sigma$ can be executed iff the first component of $\theta_c$ is $\sigma$, i.e., $\textbf{OUT}^{ctr}(\theta_c,\sigma)!$.
When this command is executed, by (\ref{Eq6}), $\theta_c\leftarrow \textbf{OUT}^{ctr}(\theta_c,\sigma)$.

\begin{figure}
	\begin{center}
		\includegraphics[width=8.6cm]{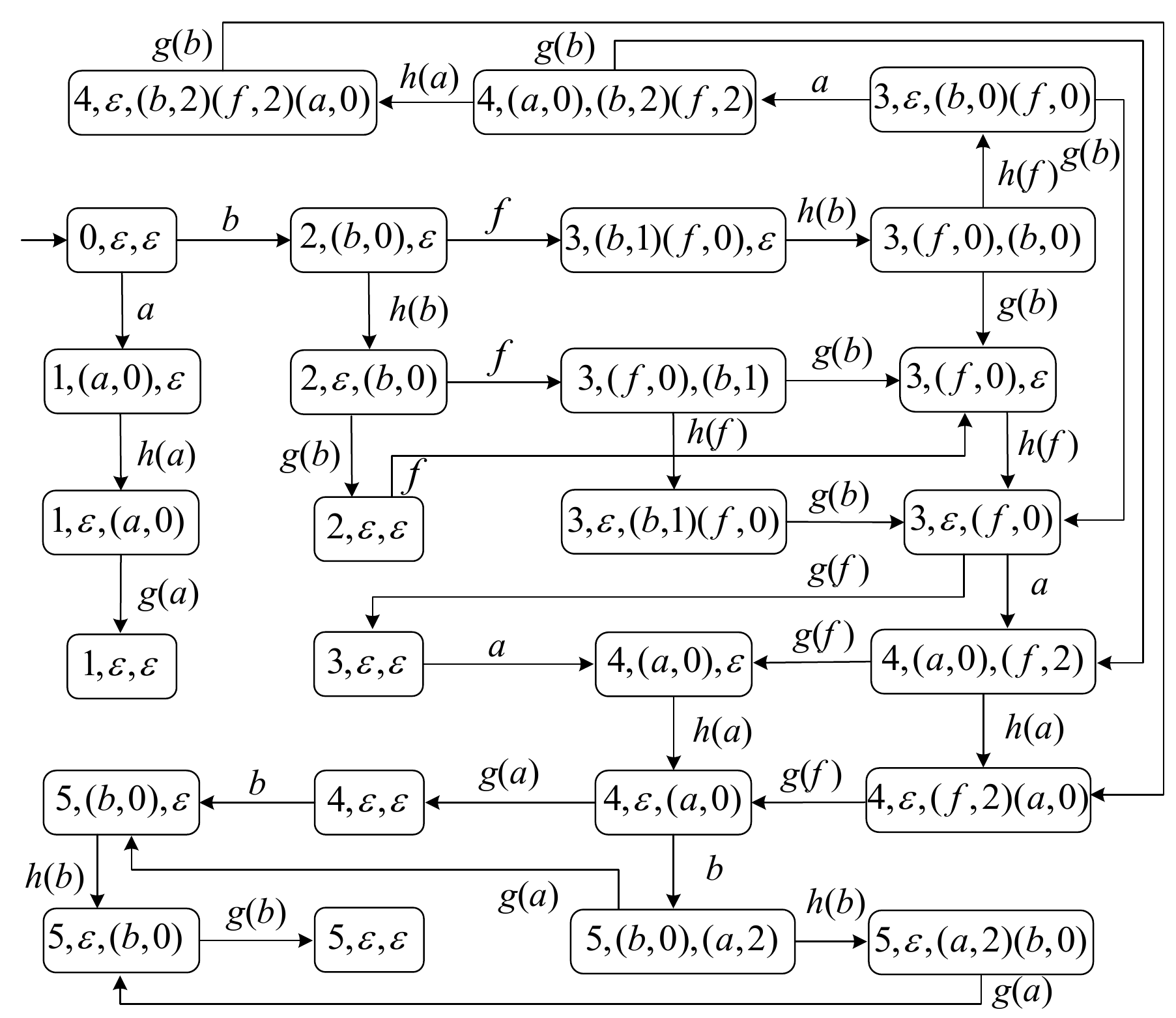}   
		\caption{The communication automaton $\tilde{G}$ in Example
\ref{Exa4}.} 
		\label{Fig3}
	\end{center}
\end{figure}


\begin{example}\label{Exa4}
We consider again Example \ref{Exa1}.
The system $G$ is given by Fig. \ref{Fig2}. 
We have $N_o=1$ and $N_c=2$. 
The communication automaton $\tilde{G}$ is shown in Fig. \ref{Fig3}.

The initial state of $\tilde{G}$ is $\tilde{q}_0=(0,\varepsilon,\varepsilon)$. 
By Fig. \ref{Fig2},  both $a$ and $b$ are active at state $0$, i.e., $\delta(0,a)=1 \wedge \delta(0,b)=2$.
When $a$ or $b$ occurs, by (\ref{Eq2}), we have $\theta_{o}=\mathbf{IN}^{obs}(\varepsilon,0,a)=(a,0)$ or $\theta'_{o}=\mathbf{IN}^{obs}(\varepsilon,0,b)=(b,0)$.
Thus,  by (\ref{Eq7}), $\tilde{\delta}(\tilde{q}_0,a)=\tilde{q}_1=(1,\theta_{o},\varepsilon)$ or $\tilde{\delta}(\tilde{q}_0,b)=\tilde{q}_2=(2,\theta'_{o},\varepsilon)$.
We next consider state  $\tilde{q}_1=(1,\theta_{o},\varepsilon)$.
Since $\theta_{o}=(a,0)$, by (\ref{Eq3}), the event occurrence of $a$ may be communicated to the agent, i.e., $\mathbf{OUT}^{obs}(\theta_{o},a)!$.
When the occurrence of $a$ is communicated,  $\theta_o \leftarrow \mathbf{OUT}^{obs}(\theta_{o},a)=\varepsilon$.
Meanwhile, following the communication of $a$, by (\ref{Eq5}), $\theta_c \leftarrow \mathbf{IN}^{ctr}(\varepsilon,a)=(a,0)$.
Thus, by (\ref{Eq8}), $\tilde{\delta}(\tilde{q}_1,h(a))=(1,\varepsilon,(a,0))$.
In this way, we can define all the transitions of $\tilde{G}$.
\end{example}

\begin{remark}
It is worth mentioning that the communication automaton $\tilde{G}$ is built for an arbitrary SAP.
This is the reason why the control channel is modeled by a string of communicated events instead of sensor activation commands issued according to a SAP.
Based on the strings of $\mathcal{L}(\tilde{G})$, we can determine which sensor activation command is taking effect.
For example, we consider a string $\mu=bh(b)f h(f)g(b)g(f)\in \mathcal{L}(\tilde{G})$ in Fig. \ref{Fig3}. 
Since $f$ occurs after $h(b)$,  $b$ is communicated before the occurrence of $f$.
However, since $g(b)$ occurs after $h(f)$,  the sensor activation decision made after the communication of $b$ is delayed and is not executed until the event occurrence of $f$ is communicated to the agent.
Therefore, we can determine that the commands taking effect right before the occurrences of $b$ and $f$ are both the initial sensor activation command.
\end{remark}

\begin{remark}
\textcolor{blue}{The computational complexity for the construction of $\tilde{G}$ is determined by the number of states of $\tilde{G}$, which is upper bounded by $|Q| \times |\Theta_o| \times |\Theta_c|$. Since $\Theta_o\subseteq (\Sigma \times [0,N_o])^{\le N}$ and $\Theta_c \subseteq (\Sigma \times [0,N_c])^{\le M}$, the number of states of $\tilde{G}$ is upper bounded by $|Q|\times |\Sigma|^{N+M} \times (N_o+1)^N \times (N_c+1)^M$.
Thus, the computational complexity for the construction of $\tilde{G}$ is polynomial with respect to (w.r.t.) $|Q|$, $|\Sigma|$, $N_o$, and $N_c$, but is exponential w.r.t. $N$ and $M$.}
\end{remark}

Given a string $\mu \in \mathcal{L}(\tilde{G})$,  let $\psi(\mu)$ and $\psi^g(\mu)$ be the strings obtained by removing all the occurrences of events in $\Sigma_h \cup \Sigma_g$ and $\Sigma \cup \Sigma_h$, respectively,  without changing the order of the remaining event occurrences in $\mu$. 
For example, let us consider $\mu=bfh(b)g(b)\in \mathcal{L}(\tilde{G})$ in Fig. \ref{Fig3}. We have $\psi(\mu)=bf$ and $\psi^g(\mu)=g(b)$.
We extend $\psi$ to a set of strings $L \subseteq \tilde{\Sigma}^*$ as: $\psi(L)=\{\psi(\mu):\mu \in L\}$.
Intuitively, for any $\mu \in \mathcal{L}(\tilde{G})$,
$\psi(\mu)$ records the string that has occurred in the plant, and $g^{-1}(\psi^g(\mu))$ tracks the string such that the sensor activation command issued after its communication is now taking effect.
For example, we consider $\mu=bh(b)f \in \mathcal{L}(\tilde{G})$ in Fig. \ref{Fig3}. 
Since $g^{-1}(\psi^g(bh(b)))=\varepsilon$, the sensor activation command being in use when $f$ occurs after $b$ is $\Omega(\varepsilon)$.

Given an SAP $\Omega$, the observation mapping $P_{\Omega}:\mathcal{L}(\tilde{G})\to \Sigma_o^*$ is iteratively defined as $P_{\Omega}(\varepsilon)=\varepsilon$, and for all $\mu,\mu\sigma \in \mathcal{L}(\tilde{G})$,
  \Lfteqn\label{Eq10}
  P_{\Omega}(\mu\sigma)= \begin{cases}
  P_{\Omega}(\mu)\sigma  & \text{if}\ \sigma \in \Omega(g^{-1}(\psi^g(\mu))) \\
  P_{\Omega}(\mu) & \mbox{otherwise.}
  \end{cases}
  \Ndeqn
For any $\mu \in \mathcal{L}(\tilde{G})$, $P_{\Omega}(\mu)$ describes what the agent will see when the occurrence of  $\psi(\mu)$ is communicated to the agent.

We show that $\tilde{G}$ does not change the system language of $G$.

\begin{proposition}\label{Prop1}
Given a networked DES $G$, let $\tilde{G}$ be constructed as described above.
Then $\psi(\mathcal{L}(\tilde{G}))=\mathcal{L}(G)$.
\end{proposition}
\begin{proof}
Please see Appendix-$A$.
\end{proof}

\begin{remark}
\textcolor{blue}{In some practical applications, there may exist communication losses between the plant and the agent. 
For example, some observable transitions may be lost when they are communicated to the agent. 
We denote the set of transitions of $G$ by $\delta_G=\{(q,\sigma,q'):\delta(q,\sigma)=q'\}$ and the set of observable transitions of $G$ by $\delta_{G,o}=\{(q,\sigma,q'):\delta(q,\sigma)=q' \wedge \sigma \in \Sigma_o\}$.
We partition $\delta_{G,o}$ into $\delta_{G,L}$ and $\delta_{G,o} \setminus \delta_{G,L}$, where $\delta_{G,L}$ is the set of transitions whose corresponding event occurrences are either observed without losses or observed with losses.
We use a new event $\sigma_{\varepsilon}$ to denote that the occurrence of $\sigma$ has been lost when it is communicated to the agent. 
To model possible observation losses, we refine the structure of $G$ by adding parallel $\sigma_{\varepsilon}$-transitions to all the transitions $(q,\sigma,q') \in \delta_{G,L}$ in $G$, and obtain $G'=(Q,\Sigma \cup \Sigma_{\varepsilon},\delta',q_0)$, where  $\Sigma_{\varepsilon}=\{\sigma_{\varepsilon}:\exists (q,\sigma,q') \in \delta_{G,L}\}$ and $\delta'=\delta_G \cup \{(q,\sigma_{\varepsilon},q'):(q,\sigma,q')\in\delta_L\}$.
Using techniques developed in this section, we can construct $\tilde{G}'$.
Although the occurrence of  $\sigma_{\varepsilon}$ cannot be sensed by the agent, we should add $t_{\downarrow}(q,\sigma)$ to $\theta_c$ and $\theta_o$  when $\sigma_{\varepsilon}$ occurs (as an event $\sigma$ has occurred but  been lost). 
Similar to (\ref{Eq10}), we can determine the observation mapping of $P'_{\Omega}$ for $G'$ using $\tilde{G}'$.
And $P'_{\Omega}$ considers observation losses besides observation delays and control delays.
In this paper, we focus on dealing with the non-deterministic delays. The formal approaches for sensor activations under communication delays and losses are beyond the scope of this paper, yet a fruitful area for future exploration.}
\end{remark}

\subsection{Delay Feasibility}

Not all arbitrary SAPs $\Omega$ can be used to activate sensors in networked DESs.
Due to observation delays and control delays, the command being in use after the occurrence of a string $s$ may not be $\Omega(s)$.
If $\sigma \in \Omega(s)$ (or $\sigma \notin \Omega(s)$), to guarantee that $\sigma$ can (or cannot) be observed by the agent after the communication of $s$, it is required that all the sensor activation commands that may be taking effect when $\sigma$ occurs after $s$ must activate (or deactivate) the sensor for $\sigma$.
In other words, the SAP should be insensitive to non-deterministic observation delays and control delays.
Additionally, to guarantee that $\Omega$ is practically feasible,  any two communicated strings that can confuse the agent must be followed by the same sensor activation mappings.

\begin{definition}\label{Def3}
\textcolor{blue}{(Delay Feasibility)} A given SAP $\Omega$ is said to be delay feasible w.r.t. $\tilde{G}$ if the following two conditions are true. 
\begin{align*}
\textup{C}1:\ &(\forall \sigma \in \Sigma)(\forall \mu\sigma \in \mathcal{L}(\tilde{G}))\sigma \in \Omega(\psi(\mu)) \Leftrightarrow\\ &\sigma \in \Omega(g^{-1}(\psi^g(\mu)));\\
\textup{C}2:\ &(\forall \sigma \in \Sigma)(\forall \mu, \mu' \in \mathcal{L}(\tilde{G}))P_{\Omega}(\mu)=P_{\Omega}(\mu') \Rightarrow \\
&[\sigma \in \Omega(\psi(\mu)) \Leftrightarrow \sigma \in \Omega(\psi(\mu'))].
\end{align*}
\end{definition}

Condition $\textup{C}1$ states that if  an SAP $\Omega$ intends to activate an event after a string, i.e., $\sigma \in \Omega(\psi(\mu))$, then this event can be activated  after this string for all possible observation delays and control delays,  i.e., $\sigma \in \Omega(g^{-1}(\psi^g(\mu)))$.
Condition $\textup{C}2$ states that any two communicated strings that are of the same observation mapping must be followed by the same activation decision for every event.
\textcolor{blue}{In principle, to check delay feasibility, we first calculate $P_{\Omega}$ from $\Omega$ and $\tilde{G}$, and then check if $\textup{C}1$ and $\textup{C}2$ hold.}

We discuss some properties of delay feasibility.
We first show that delay feasibility is closed under union.

\begin{lemma}\label{Lem1}
Given two SAPs $\Omega_1$ and $\Omega_2$, if $\Omega_1 \supseteq \Omega_2$ and $\Omega_2$ is delay feasible, then for all $\mu,\mu' \in \mathcal{L}(\tilde{G})$,
\begin{align}P_{\Omega_1}(\mu)=P_{\Omega_1}(\mu') \Rightarrow P_{\Omega_2}(\mu)=P_{\Omega_2}(\mu').\end{align}
\end{lemma}
\begin{proof}
Please see Appendix-$B$.
\end{proof}

Using the above lemma, we now prove that the union of any two delay feasible SAPs is also delay feasible.
\begin{theorem}\label{Prop2}
Given two SAPs $\Omega_1$ and $\Omega_2$, if $\Omega_1$ and $\Omega_2$ are  both delay feasible, then $\Omega_1\cup \Omega_2$ is also delay feasible.
\end{theorem}
\begin{proof}
Please see Appendix-$C$.
\end{proof}

By Theorem \ref{Prop1}, among all the delay feasible subpolicies $\{\Omega_1,\ldots,\Omega_k\}$ of an SAP $\Omega$, there always exists a supremal one
$\Omega^{\uparrow}=\Omega_1\cup\cdots\cup\Omega_k$. 
In this paper, we call $\Omega^{\uparrow}$ the $($unique$)$ maximum delay feasible subpolicy of $\Omega$.
Next, we consider the observations of an agent under a delay feasible SAP.
The following proposition states that if an SAP $\Omega$ is delay feasible, then $\Omega$ is insensitive to observation delays and control delays.

\begin{proposition}\label{Prop3}
Given a delay feasible SAP $\Omega$, we have
$[(\forall \mu \in \mathcal{L}(\tilde{G}))\psi(\mu)=s] \Rightarrow P_{\Omega}(\mu)=\theta^{\Omega}(s).$
\end{proposition}
\begin{proof}
Please see Appendix-$D$.
\end{proof}

By Proposition \ref{Prop3}, under a delay feasible SAP,  what the agent may see after the communication of a string is equal to the information mapping of this string.
Whenever a string $s$ occurs in $G$, due to observation delays, one of the strings in $\{s_{-i}:i=0,1,\ldots, N_o\}$ can be communicated to the agent. 
Thus, the set of delayed observation mappings associated with $s$ and a delay feasible SAP $\Omega$, denoted by $\Theta_{\Omega}^{N_o}(s)$, can be obtained by:
$\Theta_{\Omega}^{N_o}(s)=\{\theta^{\Omega}(s_{-i}):i=0,1,\ldots, N_o \}.$
Let us define \begin{align*}
&\mathcal{T}_{conf,G}(\Theta_{\Omega}^{N_o})=\{(q,q')\in Q\times Q:(\exists s,s'\in \mathcal{L}(G))\\
&\Theta_{\Omega}^{N_o}(s)\cap\Theta_{\Omega}^{N_o}(s')\neq \emptyset \wedge q=\delta(q_0,s)\wedge q'=\delta(q_0,s')\}
\end{align*} as the set of confusable state pairs of $G$ under $\Theta_{\Omega}^{N_o}$.
Note that $(q,q')$ is equal to  $(q',q)$ in $\mathcal{T}_{conf,G}(\Theta_{\Omega}^{N_o})$.

\begin{proposition}\label{Prop4}
Given two delay feasible SAPs $\Omega_1$ and $\Omega_2$, then
$\Omega_1 \supseteq \Omega_2 \Rightarrow \mathcal{T}_{conf,G}(\Theta_{\Omega_1}^{N_o}) \subseteq \mathcal{T}_{conf,G}(\Theta_{\Omega_2}^{N_o}).$
\end{proposition}
\begin{proof}
Please see Appendix-$E$.
\end{proof}

Proposition \ref{Prop4} exploits the fact that, the more sensors a delay feasible SAP activates, the fewer state pairs an agent confuses under observation delays and control delays.
This fact will help us find a valid minimal SAP in the next section.

\section{Minimal sensor activation policy}

In this section, we first formulate the networked sensor
activation problem to be solved.
We then discuss how to verify if a given SAP satisfies delay feasibility.
Finally, we develop procedures to compute a minimal SAP. 

\subsection{Problem Formulation}

In the above sections, $\Omega$ is language-based, i.e., the domain of $\Omega$ is the system language $\mathcal{L}(G)$.
For the sake of computational efficiency, as stated in \cite{lafortune10tac,weilin16tac},  any two  strings ending up in the same state  must implement the same observation mapping, instead of depending on the entire string as in the general case (Section II-B).
Correspondingly, the SAP can be denoted by a set of transitions $\Delta \subseteq Q \times \Sigma_o$\footnote{To fulfill the sensor activation task,  we may activate an event that is not active at the current state due to observation delays and control delays.
Thus, the SAP $\Delta$ is defined over the whole $ Q \times \Sigma_o$ instead of $\{(q,\sigma) \in Q \times \Sigma_o:\delta(q,\sigma)!\}$.} such that, given
$s\sigma \in \mathcal{L}(G)$ with $\sigma \in \Sigma$, the sensor for $\sigma$ needs to be activated  immediately after the communication of $s$ if $(\delta(q_0,s),\sigma)\in \Delta$.
Since transition-based dynamic observation is a special case of language-based dynamic observation, all results in the previous sections
still hold for systems under transition-based observation.
In the remainder of this paper, we assume that the observation mapping is transition-based.

Given a transition-based SAP $\Delta$, the observation mapping $P_{\Delta}:\mathcal{L}({\tilde{G}}) \to \Sigma_o^*$ for $\tilde{G}$  is iteratively defined as $P_{\Delta}(\varepsilon)=\varepsilon$ and, for all $\mu,\mu\sigma\in \mathcal{L}(\tilde{G})$,
\lfteqn
P_{\Delta}(\mu\sigma)=
\begin{cases}
	P_{\Delta}(\mu)\sigma & \mathrm{if}\ (\delta(q_0,g^{-1}(\psi^g(\mu))),\sigma) \in \Delta\\
	P_{\Delta}(\mu) & \mathrm{if}\ (\delta(q_0,g^{-1}(\psi^g(\mu))),\sigma) \notin \Delta.
\end{cases}
\ndeqn

By Definition \ref{Def3},  $\Delta$ is delay feasible if 
\begin{align*}
\textup{C}3:\ &(\forall \sigma \in \Sigma)(\forall \mu\sigma \in \mathcal{L}(\tilde{G}))[(\delta(q_0,\psi(\mu)),\sigma)\in \Delta \\  &\Leftrightarrow (\delta(q_0,g^{-1}(\psi^g(\mu))),\sigma)\in \Delta];\\
\textup{C}4:\ &(\forall \sigma \in \Sigma)(\forall \mu, \mu' \in \mathcal{L}(\tilde{G}))P_{\Delta}(\mu)=P_{\Delta}(\mu') \Rightarrow \\
&[(\delta(q_0,\psi(\mu)),\sigma)\in \Delta \Leftrightarrow (\delta(q_0,\psi(\mu')),\sigma)\in \Delta].
\end{align*}

\textcolor{blue}{
For the purpose of control, diagnosis, or other applications, we require that the agent is able to distinguish certain pairs of states
of $G$, which is specified as a specification condition $\mathcal{T}_{spec}\subseteq Q\times Q$. 
Note that, $\mathcal{T}_{spec}$ is user-defined and problem-dependent, and many important properties, such as observability, $K$-diagnosaiblity, and detectability can be translated into state pairs disambiguation problems \cite{lafortune07scl,yin18tac}. Thus, using $\mathcal{T}_{spec}$ to describe the observation requirement of an agent is a rather general approach.}
Formally, the problem to be solved is formulated as follows.

\begin{problem}\label{Prob1}
\textcolor{blue}{(Centralized Sensor Activation Optimization)} Suppose we are given a system $G$, a specification $\mathcal{T}_{spec}$, and a set of observable events $\Sigma_o \subseteq \Sigma$.
Assume that if the agent activates all the sensors, $\mathcal{T}_{spec}$ is satisfied, i.e., $\mathcal{T}_{conf,G}(\Theta_{\Delta^{all}}^{N_o}) \cap \mathcal{T}_{spec}=\emptyset$, where $\Delta^{all}=Q\times \Sigma_o$.
We would like to find an SAP $\Delta^*\subseteq \Delta^{all}$ such that: 1) $\Delta^*$ is delay feasible, i.e., \textup{C}3 and \textup{C}4 holds under $\Delta^*$; 
2) $\mathcal{T}_{spec}$ is satisfied under $\Delta^*$, i.e., $\mathcal{T}_{conf,G}(\Theta_{\Delta^*}^{N_o}) \cap \mathcal{T}_{spec}=\emptyset$; 3) $\Delta^*$ is minimal, i.e., there is no other $\Delta'\subset \Delta^*$ satisfying 1) and 2).
\end{problem}

\subsection{\textcolor{blue}{Checking Delay Feasibility}}

Before we formally solve Problem \ref{Prob1}, we first show how to check delay feasibility of a given transition-based SAP.
{To determine which command is taking effect, we first refine the state space of the communication automaton $\tilde{G}$ as follows:}
\begin{enumerate}
\item Replace all the transitions $(q,\sigma)$ in $G$ by $(q,g(\sigma))$ to obtain a new automaton $G^g$ (See the example discussed in Section IV-C for more details);
\item  Set  $W \leftarrow \tilde{G} || G^g$ and denote by $W=(Q_W,\tilde{\Sigma},\delta_W,q_{0,W})$.
\end{enumerate}

\begin{remark}
By 1) and 2), we have the following two results.
First, automaton $W$ generates the same language as automaton $\tilde{G}$ does, i.e., $\mathcal{L}(W)=\mathcal{L}(\tilde{G})$.
Second, for any $\mu \in \mathcal{L}(W)$, we write $\delta_W(q_{0,W},\mu)=(q,\theta_o,\theta_c, x)$ for $(q,\theta_o,\theta_c)\in \tilde{Q}$ and $x\in Q$.
Then, $x=\delta(q_0,g^{-1}(\psi^g(\mu)))$ tracks the string such that the sensor activation decision made after the communication of it is taking effect.
Thus, when an event $\sigma$ occurs at $(q,\theta_o,\theta_c, x)$, it is observable iff $(x,\sigma)\in \Delta$.
\end{remark}

For a given system $G$ and an SAP $\Delta \subseteq Q \times \Sigma_o$, we define  
\begin{align}\label{Eq14}
&\mathcal{T}_{conf,G}(P_{\Delta})=\{({q}, {q}') \in Q \times Q:(\exists \mu,\mu'\in \mathcal{L}(\tilde{G})) \notag\\
&P_{\Delta}(\mu)=P_{\Delta}(\mu')\wedge {q}=\delta(q_{0},\psi(\mu)) \wedge {q}'=\delta(q_{0},\psi(\mu'))\}
\end{align}
as the set of confusable state pairs of $G$ under $P_{\Delta}$.
To calculate $\mathcal{T}_{conf,G}(P_{\Delta})$ using $W$, we first define a transition-based SAP $\hat{\Delta}$ for $W$ as: $\hat{\Delta}=\{((q,\theta_o,\theta_c,x),\sigma)\in Q_W \times \tilde{\Sigma}: (x,\sigma) \in \Delta\}.$
Let
\begin{align}\label{Eq15}
&\mathcal{T}_{conf,W}(\theta^{\hat{\Delta}})=\{({q}_W, {q}'_W) \in Q_W \times Q_W:(\exists \mu,\mu'\in \mathcal{L}(W)) \notag\\
&\theta^{\hat{\Delta}}(\mu)=\theta^{\hat{\Delta}}(\mu')\wedge {q}_W=\delta_W(q_{0,W},\mu) \wedge {q}'_W=\delta_W(q_{0,W},\mu')\}
\end{align}
be the set of confusable state pairs of $W$ under $\theta^{\hat{\Delta}}$.\footnote{The definition of $\theta^{\hat{\Delta}}$ for a transition-based SAP $\hat{\Delta}$ can be obtained in a similar way as that of $\theta^{\Omega}$ for a language-based SAP $\Omega$, as stated in \cite{lafortune10tac}.}
An algorithm for calculating $\mathcal{T}_{conf,W}(\theta^{\hat{\Delta}})$ with polynomial time complexity with respect to the sizes of state space and the event set in $W$ was proposed in \cite{lafortune07scl}.
The following lemma will be used later.

\begin{lemma}\label{Lem2}
For any $\mu \in \mathcal{L}(\tilde{G})=\mathcal{L}(W)$, $P_{\Delta}(\mu)=\theta^{\hat{\Delta}}(\mu)$.
\end{lemma}
\begin{proof}
Please see Appendix-$F$.
\end{proof}
Using Lemma \ref{Lem2}, we can calculate $\mathcal{T}_{conf,G}(P_{\Delta})$ as follows.

\begin{proposition}\label{Prop6}
Given the set of indistinguishable state pairs  $\mathcal{T}_{conf,W}(\theta^{\hat{\Delta}})$ for $W$, we compute $\mathcal{T}_{conf,G}(P_{\Delta})$ for $G$ as follows: 
\begin{align}\label{Eq16}
&\mathcal{T}_{conf,G}(P_{\Delta})=\{(q,q')\in Q \times Q: \notag\\
&\exists  ((q,\theta_o,\theta_c,x), (q',\theta'_o,\theta'_c,x')) \in \mathcal{T}_{conf,W}(\theta^{\hat{\Delta}})\}.
\end{align}
\end{proposition}
\begin{proof}
Please see Appendix-$G$.
\end{proof}

We next discuss how to verify delay feasibility of a given transition-based SAP $\Delta$.
\begin{proposition}\label{Proposition5}
{Given system $G$, we construct $\tilde{G}$, $G^g$, and $W$ as described above. 
Then, a transition-based SAP $\Delta$ is delay feasible if and only if the following two conditions are true.
Condition 1:
\begin{align}\label{Equation14}
&(\forall(q,\theta_o,\theta_c,x)\in Q_W)(\forall \sigma\in \Sigma)\tilde{\delta}((q,\theta_o,\theta_c),\sigma)! \Rightarrow \notag \\
&[(q,\sigma)\in \Delta \Leftrightarrow (x,\sigma) \in \Delta]; 
\end{align}
Condition 2:
\begin{align}\label{Equation15}
(\forall (q,q')\in \mathcal{T}_{conf,G}(P_{\Delta}))[(q,\sigma)\in \Delta \Leftrightarrow (q',\sigma) \in \Delta].
\end{align}}
\end{proposition}
\begin{proof}
Please see Appendix-$H$.
\end{proof}

\subsection{Solutions}

With the above preparations, we are now ready to calculate a minimal transition-based SAP $\Delta^*$ such that $\Delta^*$ is delay feasible and $\mathcal{T}_{conf,G}(\Theta_{\Delta^*}^{N_o}) \cap \mathcal{T}_{spec}=\emptyset$.
In this algorithm, we first use a subroutine (Algorithm 1) to compute the maximum delay feasible subpolicy $\Delta^{\uparrow }$ of a given SAP $\Delta$, and the set of confusable state pairs $\mathcal{T}_{conf,G}(P_{\Delta^{\uparrow }})$.

\begin{algorithm}
  \caption{\textsc{Calculating} $\Delta^{\uparrow}$ and $\mathcal{T}_{conf,G}(P_{{\Delta}^{\uparrow}})$} \label{alg1}
  \LinesNumbered
  \KwIn {Automaton ${G}$ and a transition-based SAP $\Delta$\;} 
\KwOut {$\Delta^{\uparrow}$ and $\mathcal{T}_{conf,G}(P_{\Delta^{\uparrow}})$\;}
     Replace all the transitions $(q,\sigma)$ in $G$ by $(q,g(\sigma))$ to obtain $G^g$\;
Construct $\tilde{G}$ with the input $G$ as discussed in Section III\;
     Set \label{lin3} $W \leftarrow \tilde{G} || G^g$ and denote by $W=(Q_W,\tilde{\Sigma},\delta_W,q_{0,W})$\;
\Repeat{$\Delta=\Delta'$}
  {
     Set $\Delta' \leftarrow \Delta$\;
Compute $\mathcal{T}_{conf,G}(P_{\Delta})$ using (\ref{Eq16})\;
     Set \label{lin7} $\Delta \leftarrow \Delta \setminus \{(q,\sigma)\in \Delta:\textup{D}1 \lor \textup{D}2 \lor \textup{D}3\}$\;
}
\Return ${\Delta}^{\uparrow}\leftarrow \Delta$ and $\mathcal{T}_{conf,G}(P_{{\Delta}^{\uparrow}})\leftarrow \mathcal{T}_{conf,G}(P_{{\Delta}})$.
\end{algorithm}

Lines $1\sim 3$ refine the state space of $\tilde{G}$ to obtain $W$.
The repeat-until loop on Line 4 removes transitions violating $\textup{C}3$ or $\textup{C}4$, one by one, from $\Delta$.
Specifically, in Line \ref{lin7}, ``\textup{D}1'' refers to $(\exists (q,\theta_o,\theta_c,x)\in Q_W)\tilde{\delta}((q,\theta_o,\theta_c),\sigma)! \wedge (x,\sigma) \notin \Delta,$ ``\textup{D}2'' refers to $(\exists (x,\theta_o,\theta_c,q)\in Q_W)\tilde{\delta}((x,\theta_o,\theta_c),\sigma)! \wedge (x,\sigma) \notin \Delta,$ and ``\textup{D}3'' refers to
$(\exists (q,q') \in \mathcal{T}_{conf,G}(P_{\Delta}))(q',\sigma)\notin \Delta.$
If $\exists (q,\sigma)\in \Delta$ such that $\textup{D}1$, $\textup{D}2$ or $\textup{D}3$ are satisfied, by Proposition \ref{Proposition5}, delay feasibility of $\Delta$ is violated.
Thus, in Line 7, we remove $(q,\sigma)$ from $\Delta$ if it violates $\textup{D}1$, $\textup{D}2$, or $\textup{D}3$.
Algorithm 1 terminates when there is no change to $\Delta$ in an iteration.

\begin{proposition}\label{Prop7}
The output $\Delta^{\uparrow}$ of Algorithm 1 is the maximal delay feasible subpolicy of $\Delta$.
\end{proposition}
\begin{proof}
Please see Appendix-$I$.
\end{proof}


\begin{remark}\label{Rem6}
\textcolor{blue}{The number of iterations of Algorithm 1 is upper-bounded by the cardinality of the transitions of
the system $G$, i.e.,  $|Q| \times |\Sigma|$.
In each iteration, we need to calculate $\mathcal{T}_{conf,G}(P_{\Delta})$ once.
By (\ref{Eq16}), to calculate $\mathcal{T}_{conf,G}(P_{\Delta})$, we need to calculate $\mathcal{T}_{conf,W}(\theta^{\hat{\Delta}})$.
As shown in \cite{lafortune07scl}, the complexity of calculating  $\mathcal{T}_{conf,W}(\theta^{\hat{\Delta}})$ is of the order $|Q_W|^2 \times |\tilde{\Sigma}|$.
In the worst case, $|Q_W|=|Q|^2\times |\Sigma|^{N+M}\times (N_o+1)^N \times (N_c+1)^M$ and $|\tilde{\Sigma}|=3 \times |\Sigma|$.
Therefore, the worst-case complexity of Algorithm 1 is $\mathcal{O}(|Q_W|^2 \times |\tilde{\Sigma}|\times |Q| \times |\Sigma|)=\mathcal{O}(|Q|^3 \times |\Sigma|^{2\times (N+M)}\times N_o^N \times N_c^M)$.}
\end{remark}

We present the main algorithm for calculating the minimal transition-based SAP $\Delta^*$.
Let us first show how to calculate $\mathcal{T}_{conf,G}(\Theta_{\Delta}^{N_o})$ using $\mathcal{T}_{conf,G}(P_{\Delta})$. 

\begin{proposition}\label{Prop8}
Given the set of indistinguishable state pairs $\mathcal{T}_{conf,G}(P_{\Delta})$, we can calculate $\mathcal{T}_{conf,G}(\Theta_{\Delta}^{N_o})$ as follows: 
\begin{align}\label{Eq17}
\mathcal{T}_{conf,G}(\Theta_{\Delta}^{N_o})=&\{(x,x')\in Q \times Q: (\exists (q,q') \in \mathcal{T}_{conf,G}(P_{\Delta}))\notag \\
&x \in {R}^{N_o}(q)\wedge x'\in {R}^{N_o}(q')\},
\end{align}
where ${R}^{N_o}(q)=\{x \in Q:(\exists s \in \Sigma^{\le N_o})x=\delta(q,s)\}$ is the set of states in $G$ that can be
reached from state $q$ via a string $s$ with $|s|\le N_o$.
\end{proposition}
\begin{proof}
Please see Appendix-$J$.
\end{proof}

Algorithm 2 computes the optimal transition-based SAP $\Delta^*$.

\begin{algorithm}
  \caption{\textsc{Calculating} $\Delta^*$} \label{alg1}
    \LinesNumbered
  \KwIn {Automaton ${G}$ and specification $\mathcal{T}_{spec}$\;} 
\KwOut{ Minimal SAP ${\Delta}^*$\;}
   Initially, set $\Delta \leftarrow Q \times \Sigma_o$ and $D \leftarrow \emptyset$\;
\While{$\Delta \neq D$} 
{Pick a $(q,\sigma)\in \Delta \setminus D$ and set $\tilde{\Delta} \leftarrow \Delta \setminus \{(q,\sigma)\}$\;
Call Algorithm 1 with the inputs $G$ and $\tilde{\Delta}$ to compute $\tilde{\Delta}^{\uparrow}$ and $\mathcal{T}_{conf,G}(P_{\tilde{\Delta}^{\uparrow}})$\;
Calculate $\mathcal{T}_{conf,G}(\Theta_{\tilde{\Delta}^{\uparrow}}^{N_o})$ using (\ref{Eq17})\;
\eIf{$\mathcal{T}_{conf,G}(\Theta_{\tilde{\Delta}^{\uparrow}}^{N_o})\cap \mathcal{T}_{spec} = \emptyset$ }
{Set $\Delta \leftarrow \tilde{\Delta}^{\uparrow }$\;}
{Set $D\leftarrow D \cup \{(q,\sigma)\}$\;}
}
\Return ${\Delta}^* \leftarrow \Delta$.
\end{algorithm}

The while-loop on Line 2 of Algorithm 2 tries to remove transitions one by one from $\Omega$ while ensuring that $\mathcal{T}_{spec}$ is satisfied.
Specifically,  assume that a transition $(q,\sigma)$ is selected from $\Delta$ in one iteration.
By Line 3 of Algorithm 2, $\tilde{\Delta} \leftarrow \Delta \setminus \{(q,\sigma)\}$.
However, $\tilde{\Delta}$ could violate delay feasibility.
Nevertheless, by Theorem \ref{Prop2}, we can always find a supremal subplolicy $\tilde{\Delta}^{\uparrow}$ among all the subpolicies of $\tilde{\Delta}$.
If $\mathcal{T}_{conf,G}(\Theta_{\tilde{\Delta}^{\uparrow}}^{N_o})\cap \mathcal{T}_{spec} \neq \emptyset$, by Proposition \ref{Prop4}, all the delay feasible subpolicies of $\tilde{\Delta}$ violate $\mathcal{T}_{spec}$.
Thus, there is no need to further consider them in future iterations.
Correspondingly, we include $(q,\sigma)$ in $D$ in Line 9 of Algorithm 2.
If $\mathcal{T}_{conf,G}(\Theta_{\tilde{\Delta}^{\uparrow}}^{N_o})\cap \mathcal{T}_{spec} = \emptyset$, we set $\Delta \leftarrow \tilde{\Delta}^{\uparrow}$ and go to the next iteration to consider another $(q',\sigma')\in \Delta \setminus D$.
Algorithm 2 does not terminate until $\Delta=D$.

\begin{theorem}\label{Theo2}
The output of Algorithm 2 $\Delta^*$ is a solution to the sensor activation problem formulated in Problem 1. 
\end{theorem}
\begin{proof}
Please see Appendix-$K$.
\end{proof}

\begin{remark}
\textcolor{blue}{Algorithm 2 examines each transition at most once in the iterations. 
Therefore, the number of such iterations is upper-bounded by the cardinality of the transitions of
the system $G$, which is further upper bounded by $|Q| \times |\Sigma|$.
Each iteration calls Algorithm 1 once. 
As discussed in Remark \ref{Rem6}, the worst-case complexity of Algorithm 1 is $\mathcal{O}(|Q|^3 \times |\Sigma|^{2\times(N+M)}\times N_c^N \times N_o^M)$.
Thus, taking the product of $|Q| \times |\Sigma|$ and $|Q|^3 \times |\Sigma|^{2\times(N+M)} \times N_c^N \times N_o^M$, the worst-case computational complexity for solving Problem 1 is $\mathcal{O}(|Q|^4 \times |\Sigma|^{2\times(N+M)}\times N_c^N \times N_o^M)$.}
\end{remark}

\begin{figure}
	\begin{center}
		\includegraphics[width=7.6cm]{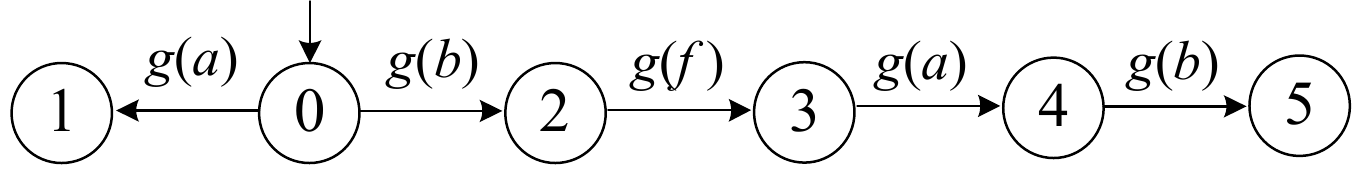}   
		\caption{Automaton $G^g$ in Example \ref{Exa5}.} 
		\label{Fig4}
	\end{center}
\end{figure}

We use the following example to illustrate how the proposed algorithms proceed.

\begin{example}\label{Exa5}
Let us again consider Example \ref{Exa1}.
The automaton $G$ with the set of observable events $\Sigma_o=\{a,b\}$ is depicted in Fig. \ref{Fig2}.
We have $t_{\downarrow}(0,a)=t_{\downarrow}(0,b)=t_{\downarrow}(3,a)=t_{\downarrow}(4,a)=2$ and $t_{\downarrow}(2,f)=1$.
We assume that $N_o=1$ and $N_c=2$, and we want to distinguish strings containing $f$ from strings without $f$ within $2$ steps.
Thus, we set $\mathcal{T}_{spec}\leftarrow\{(0,5),(1,5),(2,5)\}.$

By Line 1 of Algorithm 2, we have $\Delta \leftarrow \{0,1,\ldots,5\} \times \{a,b\}$ and $D=\emptyset$.
Next go to Line 3 of Algorithm 2. Suppose that we try to remove $(4,b)\in \Delta \setminus D$ in the first iteration.
Set $\tilde{\Delta} \leftarrow \Delta \setminus \{(4,b)\}$.
Go to Line 4 of Algorithm 2 and proceed with Algorithm 1 as follows.
By Line 1 of Algorithm 1, $G^g$ is constructed in Fig. \ref{Fig4}.
By Line 2 of Algorithm 1,  we construct $\tilde{G}$ in Fig. \ref{Fig3}.
By Line 3 of Algorithm 1, $W=\tilde{G}||G^g$.
Completing the repeat-until loop of Algorithm 1, $\tilde{\Delta}^{\uparrow} \leftarrow \{0,1,\ldots,5\} \times \{a,b\} \setminus \{(2,b),(3,b),(4,b)\}$ and $\mathcal{T}_{conf,G}(P_{\tilde{\Delta}^{\uparrow}})=\{(2,3),(4,5)\}$.
Note that we omit $(i,i)$, $i=0,1,\ldots,5$ in $\mathcal{T}_{conf,G}(P_{\tilde{\Delta}^{\uparrow}})$ for brevity.

Going to Line 5 of Algorithm 2, by (\ref{Eq17}), we can calculate
$
\mathcal{T}_{conf,G}(\Theta_{\tilde{\Delta}^{\uparrow}}^{N_o})=\{(2,3),(2,4),(3,4),(4,5)\}.$
Also, we omit $(i,i)$, $i=0,1,\ldots,5$ in $\mathcal{T}_{conf,G}(\Theta_{\tilde{\Delta}^{\uparrow}}^{N_o})$ for brevity.
It is verified that $\mathcal{T}_{conf,G}(\Theta_{\tilde{\Delta}^{\uparrow}}^{N_o})\cap \mathcal{T}_{spec} = \emptyset$.
By the if-condition of Algorithm 2,  ${\Delta} \leftarrow \tilde{\Delta}^{\uparrow}$ at the end of the first iteration.
Since $\Delta \neq D$, we go back to Line 3 of Algorithm 2.

In the second iteration of Algorithm 2, we try to remove $(0,b)$.
By Line 3 of Algorithm 2, $\tilde{\Delta} \leftarrow \Delta \setminus \{(0,b)\}$ 
We go to Line 4 of Algorithm 2 and proceed to Algorithm 1.
When completing Algorithm 1, one can check that $(0,b), (4,b) \notin \tilde{\Delta}^{\uparrow}$.
By definition, we have $(1,5)\in \mathcal{T}_{conf,G}(P_{\tilde{\Delta}^{\uparrow}})$.
By (\ref{Eq17}),  $(1,5)\in \mathcal{T}_{conf,G}(\Theta_{\tilde{\Delta}^{\uparrow}}^{N_o})$.
Therefore, $(1,5) \in \mathcal{T}_{conf,G}(\Theta_{\tilde{\Delta}^{\uparrow}}^{N_o})\cap \mathcal{T}_{spec}\neq \emptyset$.
By the if-condition of Algorithm 2 (Line 6), $D \leftarrow \{(0,b)\}$.
Finally, Algorithm 2 returns $\Delta^*=\{(0,b), (0,a), (2,a), (3,a)\}$.

We make the following comments on $\Delta^*$.
Since $(0,b)\in \Delta^*$, the agent first needs to activate the sensor for $b$.
Since $N_o=1$, once $bf$ occurs, the agent will see $b$.
Since $(0,a), (2,a), (3,a) \in \Delta^*$, the sensor for $a$ is activated at the time $a$ occurs after $bf$.
Hence, the agent will definitely see $ba$ after the occurrence of $bfab$ as $t_{\downarrow}(4,b)=2$ and the observation delays are upper bounded by 1.
The agent can infer that $f$ must have occurred after the occurrence of $bfab$.
\end{example}

\section{Application}

In this section, we first show how to apply the proposed approaches to activate sensors for fault diagnosis under delays.
Next, we use a practical example to illustrate the application.
\subsection{Delay $K$-diagnosability}
We denote by $E_f \subseteq \Sigma$  the set of fault events to be diagnosed.
The set of fault events is partitioned into $m$ mutually disjoint sets or fault types: $E_f = E_{f_1}\dot{\cup}\ldots \dot{\cup}E_{f_m}.$ 
We denote by $\mathcal{F}=\{1,\ldots,m\}$ the index set of the fault types.
 ``A fault event of type $f_i$ has occurred" means that a fault event in $E_{f_i}$ has occurred.
Let $\Psi(f_i)=\{t\sigma_f \in \mathcal{L}(H): \sigma_f \in E_{f_i}\}$ be the set of all $s \in \mathcal{L}(H)$ whose last event is a fault event of type $f_i$. 
$\mathcal{L}(H)\setminus s=\{t\in \Sigma^*:st\in\mathcal{L}(H)\}$ denotes the postlanguage of $\mathcal{L}(H)$ after $s$.
We write $f_i\in s$ if $\overline{\{s\}} \cap \Psi(f_i)\neq \emptyset$, i.e., $s$ contains a fault event of type $f_i$. 
\textcolor{blue}{A language $L \subseteq \Sigma^*$ is said to be live if whenever $s \in L$, then there exists a $\sigma \in \Sigma$ such that $s\sigma \in L$. We assume that $\mathcal{L}(G)$ is live when diagnosability is considered.
This assumption is not an actual restriction. If the modeling automaton has states without transitions
from it, we can add self-loops with a dummy unobservable event to these states,
which makes the corresponding language live.}

\begin{definition}
\textcolor{blue}{(Delay $K$-diagnosability)} A prefix-closed and live language $\mathcal{L}(G)$ is said to be delay $K$-diagnosable w.r.t. $\Theta_{\Omega}^{N_o}$, $E_f$, and $\mathcal{L}(G)$ if the following condition holds:

\begin{align}\label{Eq13}
&(\forall i \in \mathcal{F})(\forall s \in \Psi(f_i))(\forall t \in \mathcal{L}(G)\setminus s)|t|\ge K  \Rightarrow \notag \\
&(\forall u \in \mathcal{L}(G))\Theta_{\Omega}^{N_{o}}(st)\cap \Theta_{\Omega}^{N_{o}}(u) \neq \emptyset \Rightarrow f_i \in u.
\end{align}
\end{definition}

The above definition of delay $K$-diagnosability states that for any string in the system that contains any type of fault event, the agent can distinguish that string from strings without that type of fault event within $K$ steps under the delayed observation mapping $\Theta_{\Omega}^{N_o}$.
When $N_o=0$, delay $K$-diagnosability reduces to $K$-diagnosability  \cite{dallal14tac,yin18tac,yin19ac}.

Next, we show that the delay $K$-diagnosability problem can also be formulated as a sate pairs disambiguation problem.
Let $H=(X,\Sigma,\xi,x_0)$ be a given system.
As in \cite{yin18tac}, to track the number of event occurrences since each type of fault has occurred,  we refine the state space of $H$ by constructing a new automaton ${G}=({Q},\Sigma,{\delta}, {q}_{0})$, where ${Q} \subseteq X \times\{-1,0,1,\ldots,K\}^m$ is the set of states, ${q}_{0}=(x_0,\underbrace{-1,\ldots,-1}_m)$ is the initial state, and the transition function ${\delta}:{Q} \times \Sigma \rightarrow {Q}$ is defined as follows: for any ${q}=(x,n_1,\ldots,n_m)\in {Q}$ and any $\sigma \in \Sigma$, we have ${\delta}({q},\sigma)=(\xi(x,\sigma),n'_1,\ldots, n'_m),$
where for each $i\in \{1,\ldots,m\}$,
\Lfteqn\label{Eq1-1}
n'_i=
\begin{cases}
	n_i & \mathrm{if}\ [n_i=K] \lor [n=-1 \wedge \sigma \notin  E_{f_i}] \\
	n_i+1 & \mathrm{if}\ [0 \le n_i <K]  \lor [n_i=-1 \wedge \sigma \in E_{f_i}].
\end{cases}
\Ndeqn


To preserve delay $K$-diagnosability, the agent must sufficiently activate sensors to distinguish state pairs included in a specification that is defined as:
\begin{align}
\mathcal{T}_{spec}=&\{(q=(x,n_1,\ldots,n_m),q'=(x',n'_1,\ldots,n_m')) \in {Q} \times {Q}: \notag\\
&(\exists i \in \{1,\ldots, m\})n_i=K \wedge n_i'=-1\}.
\end{align}
We have that $(q,q')$ is equal to  $(q',q)$ in $\mathcal{T}_{spec}$. 
All of the above lead to the following proposition.
\begin{proposition}\label{Proposition8}
A prefix-closed and live language $\mathcal{L}(G)$ is delay $K$-diagnosable w.r.t. $\Theta_{\Omega}^{N_o}$, $E_f$, and $\mathcal{L}(G)$ iff $\mathcal{T}_{conf,{G}}(\Theta_{\Omega}^{N_o}) \cap \mathcal{T}_{spec}=\emptyset$.
\end{proposition}
\begin{proof}
Please see Appendix-$L$.
\end{proof}

\begin{remark}
\textcolor{blue}{Proposition \ref{Proposition8} shows that delay $K$-diagnosability is actually an instance of state pairs disambiguation property.
Therefore, the proposed framework is applicable for solving the dynamic sensor activation problems for the purposes of fault diagnosis of networked DESs.
Furthermore, due to the generality of $\mathcal{T}_{spec}$ \cite{lafortune07scl,yin18tac}, the proposed framework is also applicable for solving the dynamic sensor activation problems for other state pairs disambiguation properties (not confined to just the fault diagnosis problem), such as, fault prognosis problem and supervisory control problem, in the context of networked DESs.}
\end{remark}

\subsection{\textcolor{blue}{Illustrative Example}}

Next, we use a practical example to show the application of the proposed framework.
We consider a production line that is serviced by a robot.
The production line consists of two machines, named  Machines $A$ and $B$, and a conveyor belt.
Each production part or piece needs to be subsequently processed at  Machines $A$ and $B$.
When a part arrives, the robot places it on the Machine $A$.
After this part is processed by Machine $A$, the robot places it on the conveyor belt.
And when it arrives at Machine $B$, the robot takes it to Machine $B$ for completing the remaining process.
For simplicity, we make the following assumptions: 1) The robot needs to serve a machine while the machine is busy, but it cannot serve both machines together;
2) The number of parts waiting to be processed at Machine $B$ will never exceed 1.
The production line may get stuck and fail to work due to machine or conveyor belt failures.
Here we consider only the conveyor belt failure.
To diagnose the conveyor belt failure, we next show how to dynamically activate sensors under  delays.
The event set for this example is $\Sigma=\{a,b,f,c,d\}$, where
\begin{itemize}
\item{
$a$: A part arrives at Machine $A$; the robot takes it to Machine $A$;  Machine $A$ starts to work on this part.}
\item{
$b$: Machine $A$ completes its process; the robot puts the processed part on the conveyor belt.}
\item{
$f$:  Machine $A$ completes its process; the robot puts the processed part on the conveyor belt; but the conveyor belt has been broken down.}
\item{$c$: Robot 2 takes a part from the conveyor belt to Machine $B$;  Machine $B$ starts to process it.}
\item{$d$: Machine $B$ completes the process; the robot takes the processed part away from  Machine $B$.
}
\end{itemize}

The state of the production line can be defined as a vector $\bar{x}_k=[q_1,q_2,q_3,q_4]$, $k=0,1,\ldots,5$, where $q_{1}$ takes 0 or 1, corresponding to Machine $A$ being idle or busy; $q_{2}$ takes 0 or 1, corresponding to Machine $B$ being idle or busy; $q_{3}$ also takes 0 or 1, corresponding to the conveyor belt being working well or broken down; $q_{4}$ is number of parts waiting  at Machine $B$, which is constrained to the values $\{0,1\}$.

\begin{figure}
	\begin{center}
		\includegraphics[width=6.8cm]{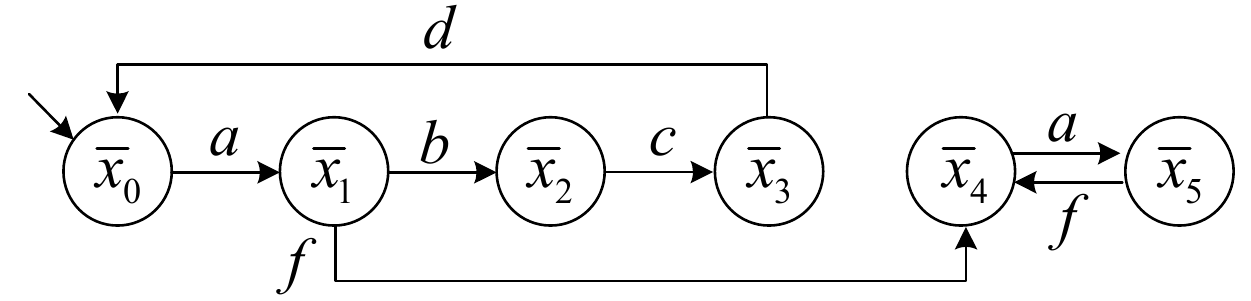}   
		\caption{System $G$ of the production line.} 
		\label{Fig5}
	\end{center}
\end{figure}

The automaton for the system is depicted in Fig. \ref{Fig5}.
At the initial state $\bar{x}_0=[0,0,0,0]$, both
Machines $A$ and $B$ are idle, the conveyor belt is working well, and there is no part waiting to be processed at Machine $B$. 
Only event $a$ could occur at $\bar{x}_0$, which corresponds to a part arriving at Machine $A$ and Machine $A$ starting to process this part.
Upon the occurrence of $a$, the system moves to $\bar{x}_1=[1,0,0,0]$. 
In state $\bar{x}_1$, $b$ could occur, which corresponds to Machine $A$ completing its process and the robot placing the processed part on the conveyor belt.
The conveyor belt may get stuck.
Thus,  $f$ could also occur in state $\bar{x}_1$.

If $b$ occurs in state $\bar{x}_1$, the system moves to state $\bar{x}_2=[0,0,0,1]$. 
Since the robot can only serve a machine at a time, only $c$ could occur in state $\bar{x}_2$, which corresponds to Machine $B$ starting to process the newly arrived part. 
Upon the occurrence of $c$, the system moves to state $\bar{x}_3=[0,1,0,0]$.
When Machine $B$ completes its process, upon the occurrence of $d$, the system moves from state $\bar{x}_3$ to  state $\bar{x}_0$ and makes state transitions as mentioned above.

On the other hand, if $f$ occurs in state $\bar{x}_1$, the system moves to state $\bar{x}_4=[0,0,1,0]$. 
In state $\bar{x}_4$, since the conveyor belt does not work, $a$ is the only event could occur, which corresponds to a new part arriving at Machine $A$, and Machine $A$ starting to process this new part.
Upon the occurrence of $a$, the system moves to state $\bar{x}_5=[1,0,1,0]$. 
When Machine $A$ completes its process, the system moves back to $\bar{x}_4$  upon the occurrence of $f$, and repeats above processes again and again.


We assume that events $a,b,c,d$ are observable, and $f$ is unobservable, that is, $\Sigma_o=\{a,b,c,d\}$ and $\Sigma_{uo}=\{f\}$.
Since it takes time to  complete a process, let $t_{\downarrow}(\bar{x}_0,a)=t_{\downarrow}(\bar{x}_2,c)=t_{\downarrow}(\bar{x}_4,a)=3$. Additionally, let $t_{\downarrow}(\bar{x}_1,b)=t_{\downarrow}(\bar{x}_1,f)=t_{\downarrow}(\bar{x}_5,f)=t_{\downarrow}(\bar{x}_3,d)=1$.
Assume that both the observation  and control delays are upper bounded by 2, i.e., $N_o=N_c=2$.

Suppose that we want to diagnose if the conveyor belt has been broken down ($f$ occurs or not).
If the sensors for $a$ and $c$ keep working since the system starts running, $\mathcal{L}(G)$ is delay $2$-diagnosable. 
That is because that we can always observe $aa$ after the occurrence $afaf$, which differs from observations of any strings that do not contain $f$.

By using Algorithms 1 and 2, one can calculate a minimal SAP
$
\Delta^*=\{(\bar{x}_0,a),(\bar{x}_1,a), (\bar{x}_2,a), (\bar{x}_3, a), (\bar{x}_4,a), (\bar{x}_5,a)\} \cup \{(\bar{x}_0,b),(\bar{x}_1,b), (\bar{x}_2,b), (\bar{x}_4,b), (\bar{x}_5,b)\}
$
such that $\mathcal{L}(G)$ is delay $2$-diagnosable under $\Delta^*$.
By $\Delta^*$, we need to first activate the sensor for $a$ at the time it occurs.
Because otherwise, we cannot distinguish $afaf$ from $\varepsilon$.
Meanwhile, to distinguish $afaf$ and $abcda$, $\Delta^*$ actives the sensor for $b$ when it occurs (but deactives the sensors for $c$ and $d$ for minimizing sensor activations.).
To ensure delay feasibility of $\Delta^*$, by $\textup{C}3$ and $\textup{C}4$, we have that  $(\bar{x}_i,a)\in \Delta^*$ for $i=0,1,\ldots,5$, and $(\bar{x}_j,b)\in \Delta^*$ for $j=0,1,2,4,5$.

\section{\textcolor{blue}{Extension}}

\begin{figure}
	\begin{center}
		\includegraphics[width=7.0cm]{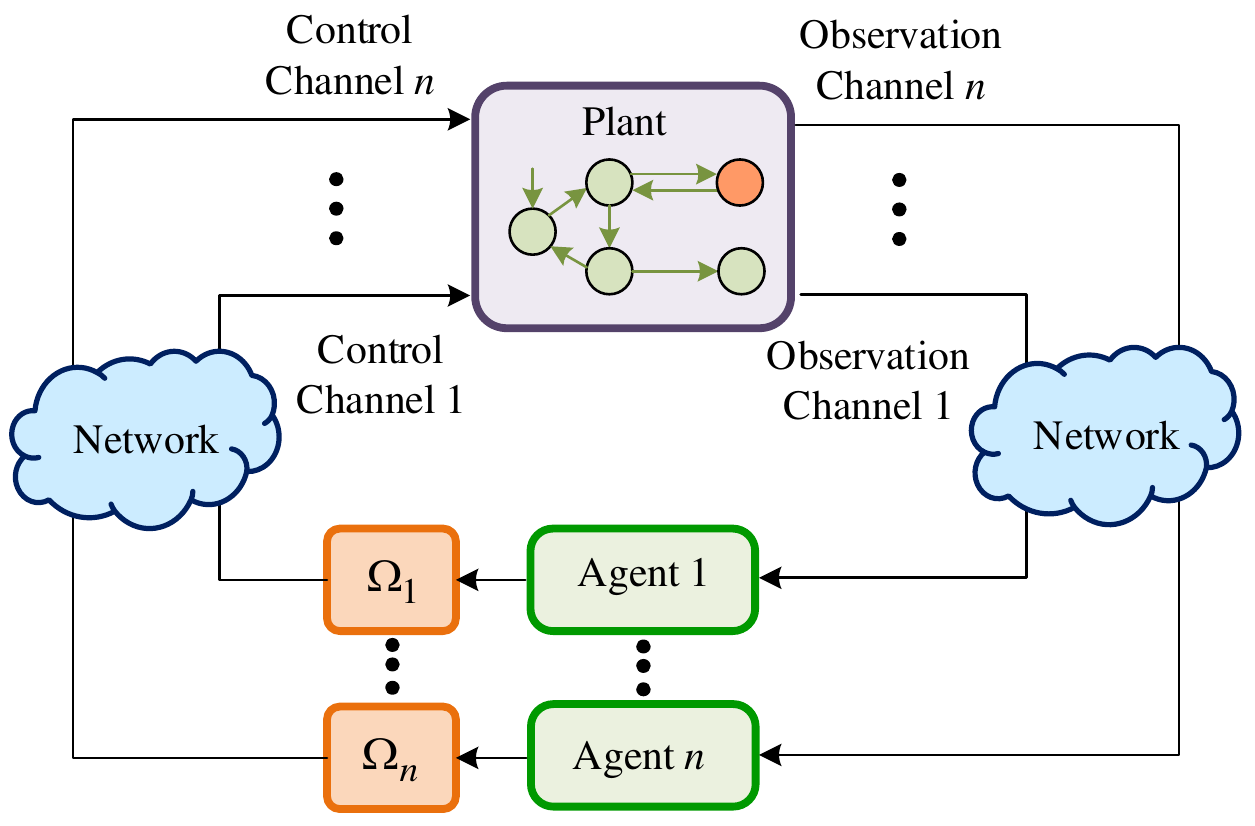}   
		\caption{Decentralized dynamic sensor activations for networked DESs.} 
		\label{Fig6}
	\end{center}
\end{figure}

In this section, we briefly discuss how to extend the proposed framework to a decentralized observation setting.
In many applications, such as cyber-physical systems, the information structure of the system is naturally decentralized, as
components of the system are distributed. 
Decentralized decision making, where several agents jointly make sensor activation decisions based on their own observations, is an efficient way
to observe these large-scale networked systems.
As shown in Fig. \ref{Fig6}, there is a set of $n$ partial-observation  agents, and all the local agents work independently, i.e., there are no communications among them. 
Events occurring in the plant are communicated to local agent $i$ over observation channel $i$, and the local agent $i$ sends sensor activation decisions to the actuator of the plant based on its SAP over control channel $i$.
Communication delays exist in each communication channel.
We assume that observation delays for observation channel $i$ are upper bounded by $N_{o,i}$ units of time, and control delays for control channel $i$ are upper bounded by $N_{c,i}$ units of time.

We denote by $I=\{1,2,\ldots,n\}$ the set of all the agents.
For each local agent $i\in I$, the set of observable events is denoted by $\Sigma_{o,i} \subseteq \Sigma_o$, and the set of controllable events is denoted by $\Sigma_{c,i}\subseteq \Sigma_c$, where $\Sigma_o=\cup_{i=1}^n\Sigma_{o,i}$ and $\Sigma_c=\cup_{i=1}^n\Sigma_{c,i}$. 
$\Sigma_{uo}=\Sigma \setminus \Sigma_o$ is the set of events that are unobservable to all the local agents.


Let $\Delta_i$ be the transition-based SAP for agent $i$.
We define $\bar{\Delta}=[\Delta_1,\ldots,\Delta_n]$.
Then, $\bar{\Delta}' \subseteq \bar{\Delta}''$ means that $(\forall i\in I)\Delta_i'\subseteq \Delta_i''$; and $\bar{\Delta}' \subset \bar{\Delta}''$ means that $[\bar{\Delta}' \subseteq \bar{\Delta}'']\wedge [(\exists i\in I)\Delta_i'\subset \Delta_i'']$.
$\bar{\Delta}$ is said to be delay feasible if $\Delta_i$ is delay feasible for all $i=1,\ldots,n$.
Let $\theta_i^{\Delta_i}:\mathcal{L}(G)\rightarrow \Sigma_{o,i}^*$ be the information mapping for agent $i$.
For any $i\in I$ and any $s\in\mathcal{L}(G)$, let $\Theta_{\Delta_i}^{N_{o,i}}(s)=\{\theta_i^{\Delta_i}(s_{-i}):i=0,1,\ldots, N_{o,i}\}$ be the delayed observation mapping associated with $s \in \mathcal{L}(G)$ and $\Delta_i$, $i=1,\ldots,n$.
Let $\Theta^{\bar{\Delta}}=[\Theta_{\Delta_1}^{N_{o,1}},\ldots, \Theta_{\Delta_n}^{N_{o,n}}]$.
Let us define \begin{align*}
&\mathcal{T}_{conf,G}(\Theta^{\bar{\Delta}})=\{(q,q_1,\ldots, q_n)\in \underbrace{Q \times \cdots \times Q}_{n+1}:\\
&(\exists s,s_1,\ldots,s_n\in \mathcal{L}(G))[q=\delta(q_0,s)] \wedge \\
&[(\forall i\in I)q_i=\delta(q_0,s_i)\wedge \Theta_{\Delta_i}^{N_{o,i}}(s)\cap \Theta_{\Delta_i}^{N_{o,i}}(s_i) \neq \emptyset]\}
\end{align*} as the set of confusable vectors of states of $G$ under $\Theta_{\bar{\Delta}}^{N_o}$.

For the purpose of decentralized  control, diagnosis, or other applications, the agents must sufficiently activate sensors
to distinguish vectors of states of $G$.
We denote all these vectors of states by a specification $\mathcal{T}'_{spec} \subseteq \underbrace{Q \times \cdots \times Q}_{n+1}$.
We require that all the $(q,q_1,\ldots,q_n)\in \mathcal{T}'_{spec}$ are distinguishable from the viewpoint of the agents, i.e.,
\begin{align}\label{Eq20}
\mathcal{T}_{conf,G}(\Theta^{\bar{\Delta}}) \cap \mathcal{T}'_{spec}=\emptyset.
\end{align}
We say that $\bar{\Delta}$ satisfies $\mathcal{T}'_{spec}$ if (\ref{Eq20}) holds.

Since we focus on the centralized sensor activation problem, the formal algorithms for the computation of  $\mathcal{T}_{conf,G}(\Theta^{\bar{\Delta}})$ and the verification of (\ref{Eq20}) are beyond the scope of this paper.


\begin{problem}\label{Prob2}
\textcolor{blue}{(Decentralized Sensor Activation Optimization)} Suppose we are given a system $G$, a specification $\mathcal{T}'_{spec}$, and a set of agents.
Assume that if all the agents activates all the sensors, $\mathcal{T}'_{spec}$ is satisfied, i.e., $\mathcal{T}_{conf,G}(\Theta^{{\bar{\Delta}^{all}}}) \cap \mathcal{T}_{spec}=\emptyset$, where $\bar{\Delta}^{all}=[{\Delta}_1^{all},\ldots,{\Delta}_n^{all}]$ with ${\Delta}_i^{all}=Q \times \Sigma_{o,i}$.
We would like to find an SAP $\bar{\Delta}^*=[{\Delta}^*_1,\ldots,{\Delta}^*_n]\subseteq \bar{\Delta}^{all}$ such that:
1)  $\bar{\Delta}^*$ is delay feasible, i.e., ${\Delta}^*_i$ satisfies \textup{C}3 and \textup{C}4 for $i=1,\ldots, n$; 2) $\mathcal{T}'_{spec}$ is satisfied under $\bar{\Delta}^*$, i.e., $\mathcal{T}_{conf,G}(\Theta^{\bar{\Delta}^*}) \cap \mathcal{T}'_{spec}=\emptyset$; 3) $\bar{\Delta}^*$ is minimal, i.e., there is no other $\bar{\Delta}'\subset \bar{\Delta}^*$ satisfying 1) and 2).
\end{problem}

The following proposition will be used later.
\begin{proposition}\label{Proposition9}
Let $\bar{\Delta}'$ and $\bar{\Delta}''$ be two delay feasible SAPs and $\bar{\Delta}' \subseteq \bar{\Delta}''$.
If $\bar{\Delta}'$ satisfies $\mathcal{T}'_{spec}$, then $\bar{\Delta}''$ satisfies $\mathcal{T}'_{spec}$.
\end{proposition}
\begin{proof}
Please see Appendix-$M$.
\end{proof}

Building on the results of Section IV and exploiting Proposition \ref{Proposition9}, we propose the Algorithm 3 for solving Problem \ref{Prob2}.

\begin{algorithm}
  \caption{\textsc{Calculating} $\bar{\Delta}^*$} \label{alg1}
    \LinesNumbered
  \KwIn {Automaton ${G}$ and specification $\mathcal{T}'_{spec}$\;} 
\KwOut{ Minimal SAP $\bar{\Delta}^*$\;}
   Initially, set $I=\{1,\ldots,n\}$. Set $\Delta_i \leftarrow Q \times \Sigma_{o,i}$ and $D_i \leftarrow \emptyset$ for all $i=1,\ldots,n$. Set $\bar{\Delta}=[\Delta_1,\ldots,\Delta_n]$\;
\While{$I \neq \emptyset$} 
{
Pick an $i\in I$ randomly\;
\While{$\Delta_i \neq D_i$} 
{Pick a $(q,\sigma)\in \Delta_i \setminus D_i$ and set $\tilde{\Delta}_i \leftarrow \Delta_i \setminus \{(q,\sigma)\}$\;
Call Algorithm 1 with the inputs $G$ and $\tilde{\Delta}_i$ to compute the maximum delay feasible $\tilde{\Delta}_i^{\uparrow}$\;
Set $\bar{\Delta}^{\uparrow}\leftarrow [\Delta_1,\ldots,\Delta_{i-1}, \tilde{\Delta}_i^{\uparrow},\Delta_{i+1},\ldots,\Delta_n]$\;
\eIf{$\mathcal{T}_{conf,G}(\Theta^{\bar{\Delta}^{\uparrow}}) \cap \mathcal{T}'_{spec}=\emptyset$}
{Set $\Delta_i \leftarrow \tilde{\Delta}^{\uparrow}_i$\;}
{Set $D_i \leftarrow D_i \cup \{(q,\sigma)\}$\;}
}
Set $I\leftarrow I \setminus \{i\}$\;
}
\Return $\bar{\Delta}^* \leftarrow \bar{\Delta}=[\Delta_1,\ldots,\Delta_n]$.
\end{algorithm} 

The correctness of Algorithm 3 is established in the
following theorem
\begin{theorem}\label{Theo3}
The output $\bar{\Delta}^*$ of Algorithm 3 is a solution to Problem 2.
\end{theorem}
\begin{proof}
Please see Appendix-$N$.
\end{proof}

\begin{remark}
\textcolor{blue}{The number of iterations of Algorithm 3 is upper bounded by the cardinality of the transitions of the system times the number of agents, i.e., $n\times |Q|\times |\Sigma|$. Each iteration
calls Algorithm 1 once. 
The worst-case complexity of Algorithm 1 is $\mathcal{O}(|Q|^3 \times |\Sigma|^{2\times(N+M)}\times N_c^N \times N_o^M)$.
In all, the worst-case complexity of Algorithm 3 is of
the order of $n\times |Q|\times |\Sigma|$ times the maximum of $|Q|^3 \times |\Sigma|^{2\times(N+M)}\times N_c^N \times N_o^M$ and
the complexity for verifying $\mathcal{T}_{conf,G}(\Theta^{\bar{\Delta}^{\uparrow}}) \cap \mathcal{T}'_{spec}$.
}
\end{remark}

\section{Conclusion}

In this paper, we have considered the sensor activation problem for networked DESs, where the sensor systems communicate with the agent over a shared network subject to observation delays and control delays.
A novel framework for sensor activations under communication delays has been established.
Under this framework, we have proposed the definition of delay feasibility of an SAP.  
We have shown that a delay feasible SAP can be used to dynamically activate sensors, even if there exist non-deterministic control delays and observation delays.
In this context, we have proven that SAPs possess a monotonicity property whenever they satisfy the delay feasibility requirement.
We have also shown that a unique maximum delay feasible subpolicy of a given SAP always exists.
Using these results, we have developed a set of algorithms for the optimization of dynamic sensor activation in networked DESs.
We have shown how to apply the proposed procedures to optimize sensor activations for the purpose of fault diagnosis.
A practical example is also provided to illustrate this application.
Finally, we have extended the proposed framework to a decentralized observation setting.

\bibliographystyle{IEEEtran}     
\bibliography{articles}

\clearpage

\appendix

\subsection{Proof of Proposition \ref{Prop1}}
\begin{proof}
We first prove $\psi(\mathcal{L}(\tilde{G})) \subseteq \mathcal{L}(G)$.
For any $\mu \in \mathcal{L}(\tilde{G})$, we write $\tilde{\delta}(\tilde{q}_0,\mu)=(q,\theta_o,\theta_c)$.
We next prove $q=\delta(q_0,\psi(\mu))$ by induction on the finite length of strings in $\mathcal{L}(\tilde{G})$.
Since $\psi(\varepsilon)=\varepsilon$ and $\tilde{q}_0=(q_0,\varepsilon,\varepsilon)$, the base case is trivially true.
The induction hypothesis is that for any $\mu \in \mathcal{L}(\tilde{G})$ such that $|\mu|\le n$, if $\tilde{\delta}(\tilde{q}_0,\mu)=(q,\theta_o,\theta_c)$, then $q=\delta(q_0,\psi(\mu))$.
We next prove the same is also true for $\mu e \in \mathcal{L}(\tilde{G})$ with $|\mu|=n$.
We write $\tilde{\delta}(\tilde{q}_0,\mu e)=(q',\theta'_o,\theta'_c)$.
By definition,  $e\in \Sigma$ or $e\in \Sigma_h \cup \Sigma_g$.
If $e \in \Sigma$, by (\ref{Eq7}), $q'=\delta(q,e)$.
By the induction hypothesis, $q=\delta(q_0,\psi(\mu))$.
Since $\psi(\mu e)=\psi(\mu)e$, we have  $q'=\delta(\delta(q_0,\psi(\mu)), e)=\delta(q_0,\psi(\mu e))$.
Otherwise, if $e \in \Sigma_h \cup \Sigma_g$, by (\ref{Eq8}) and (\ref{Eq9}),  $q'=q$.
By the induction hypothesis, $q=\delta(q_0,\psi(\mu))$.
Since $\psi(\mu e)=\psi(\mu)$, $q'=q=\delta(q_0,\psi(\mu))=\delta(q_0,\psi(\mu e))$.
Therefore, $\delta(q_0,\psi(\mu))$ is defined for any $\mu \in \mathcal{L}(\tilde{G})$,  which  implies $\psi(\mathcal{L}(\tilde{G})) \subseteq \mathcal{L}(G)$.

We next prove $\psi(\mathcal{L}(\tilde{G})) \supseteq \mathcal{L}(G)$.
For any $s=\sigma_1\cdots\sigma_k \in \mathcal{L}(G)$, by (\ref{Eq7}), (\ref{Eq8}), and (\ref{Eq9}), it can be easily verified that $\mu=\sigma_1h(\sigma_1)g(\sigma_1)\cdots \sigma_kh(\sigma_k)g(\sigma_k)\in \mathcal{L}(\tilde{G})$.
By the definition of $\psi(\cdot)$, $\psi(\mu)=s$.
Thus, we have $\psi(\mathcal{L}(\tilde{G})) \supseteq \mathcal{L}(G)$.
\end{proof}

\subsection{Proof of Lemma \ref{Lem1}}

\begin{proof}
The proof is by contradiction.
Suppose that $P_{\Omega_1}(\mu)= P_{\Omega_1}(\mu')$ and $P_{\Omega_2}(\mu)\neq P_{\Omega_2}(\mu')$.
Without loss of generality, let $\mu$ and $\mu'$ be the shortest such strings, i.e., 
\begin{align}\label{Eq11}
&(\forall \nu, \nu'\in \mathcal{L}(\tilde{G})) [P_{\Omega_1}(\nu)=P_{\Omega_1}(\nu') \wedge  P_{\Omega_2}(\nu) \neq P_{\Omega_2}(\nu')] \notag \\
&\Rightarrow |\mu|+|\mu'| \le |\nu|+|\nu'|.
\end{align}
Clearly, $\mu \neq \varepsilon$.
Because otherwise if $\mu=\varepsilon$, for any $\mu'\in \mathcal{L}(\tilde{G})$ with $P_{\Omega_1}(\mu)=P_{\Omega_1}(\mu')=\varepsilon$, by $\Omega_1\supseteq \Omega_2$, $P_{\Omega_2}(\mu)=P_{\Omega_2}(\mu')=\varepsilon$.
Similarly,  $\mu' \neq \varepsilon$.
Since $\mu,\mu'\neq\varepsilon$, we write $\mu=s\sigma$ and $\mu'=s'\sigma'$ for some $\sigma, \sigma' \in \tilde{\Sigma}$.
We next prove $P_{\Omega_1}(s\sigma)=P_{\Omega_1}(s)\sigma$ by contradiction.

Suppose that $P_{\Omega_1}(s\sigma)=P_{\Omega_1}(s)$.
By the definition of $P_{\Omega_1}(\cdot)$,  (i) $\sigma \in \Sigma_h\cup\Sigma_f$ or (ii) $\sigma \in \Sigma \wedge \sigma \notin \Omega_1(g^{-1}(\psi^g(s)))$.
If (i) is true, by the definition of $P_{\Omega_2}(\cdot)$, $P_{\Omega_2}(s\sigma)=P_{\Omega_2}(s)$.
If (ii) is true, since $\Omega_1 \supseteq \Omega_2$, $\sigma \notin \Omega_2(g^{-1}(\psi^g(s)))$.
By the definition of $P_{\Omega_2}(\cdot)$, we also have $P_{\Omega_2}(s\sigma)=P_{\Omega_2}(s)$.
Thus, for both (i) and (ii), we have $P_{\Omega_2}(\mu)=P_{\Omega_2}(s\sigma)=P_{\Omega_2}(s)$.
Since $P_{\Omega_1}(\mu)=P_{\Omega_1}(\mu') \wedge  P_{\Omega_2}(\mu)\neq P_{\Omega_2}(\mu')$, we have $P_{\Omega_1}(s)=P_{\Omega_1}(\mu') \wedge  P_{\Omega_2}(s)\neq P_{\Omega_2}(\mu')$.
Since $|s|+|\mu'|<|\mu|+|\mu'|$, it contradicts (\ref{Eq11}).
Therefore, $P_{\Omega_1}(s\sigma)=P_{\Omega_1}(s)\sigma$.
Similarly,  we can prove that $P_{\Omega_1}(s'\sigma')=P_{\Omega_1}(s')\sigma'$.

Since $P_{\Omega_1}(s\sigma)=P_{\Omega_1}(s)\sigma$ and $P_{\Omega_1}(s'\sigma')=P_{\Omega_1}(s')\sigma'$ and $P_{\Omega_1}(s\sigma) = P_{\Omega_1}(s'\sigma')$,  we have $P_{\Omega_1}(s)=P_{\Omega_1}(s')$ and $\sigma=\sigma'$.
Thus,  $P_{\Omega_2}(s)=P_{\Omega_2}(s')$, because otherwise, $P_{\Omega_2}(s) \neq P_{\Omega_2}(s')$ contradicts (\ref{Eq11}).
Since $\Omega_2$ is delay feasible, by \textup{C}2, $\sigma \in \Omega_2(\psi(s))\Leftrightarrow \sigma \in \Omega_2(\psi(s'))$.
By \textup{C}1, $\sigma \in \Omega_2(g^{-1}(\psi^g(s)))\Leftrightarrow \sigma \in \Omega_2(g^{-1}(\psi^g(s')))$.
Thus, by $P_{\Omega_2}(s)=P_{\Omega_2}(s')$,  we have $P_{\Omega_2}(s\sigma)=P_{\Omega_2}(s'\sigma)$, which contradicts $P_{\Omega_2}(\mu) \neq P_{\Omega_2}(\mu')$.
\end{proof}

\subsection{Proof of Theorem \ref{Prop2}}

\begin{proof}
The proof is by contradiction.
Let $\Omega=\Omega_1 \cup \Omega_2$.
Suppose that $\Omega$ is not delay feasible, i.e., $\neg \textup{C1}\lor \neg \textup{C2}$.

If $\neg \textup{C1}$, there exist $\sigma \in \Sigma$ and $\mu\sigma \in \mathcal{L}(\tilde{G})$ such that (i) $\sigma \in \Omega(\psi(\mu)) \wedge \sigma \notin \Omega(g^{-1}(\psi^g(\mu)))$ or (ii) $\sigma \notin \Omega(\psi(\mu)) \wedge \sigma \in \Omega(g^{-1}(\psi^g(\mu)))$.
Let us consider the case of (i).
Since $\Omega=\Omega_1 \cup \Omega_2$ and $\sigma \in \Omega(\psi(\mu))$, there exists $i\in \{1,2\}$ such that $\sigma  \in \Omega_i(\psi(\mu))$.
Moreover, since  $\Omega=\Omega_1 \cup \Omega_2$ and  $\sigma \notin \Omega(g^{-1}(\psi^g(\mu)))$,  $\sigma \notin \Omega_i(g^{-1}(\psi^g(\mu)))$, which contradicts that $\Omega_i$ is delay feasible.
Similar to the case of (i), we can draw the conclusion that  $\Omega_i$ is not delay feasible, if case (ii) holds.

Otherwise, if $\neg \textup{C2}$, there exist $\sigma \in \Sigma$ and $\mu, \mu'\in \mathcal{L}(\tilde{G})$ such that (a) $P_{\Omega}(\mu)=P_{\Omega}(\mu')$ and $\sigma \in \Omega(\psi(\mu)) \wedge \sigma \notin \Omega(\psi(\mu'))$ or (b) $P_{\Omega}(\mu)=P_{\Omega}(\mu')$ and $\sigma \notin \Omega(\psi(\mu)) \wedge \sigma \in \Omega(\psi(\mu'))$.
We consider the case of (a).
Since $\Omega=\Omega_1 \cup \Omega_2$ and $\sigma \in \Omega(\psi(\mu))$,  there exists $i \in \{1,2\}$ such that $\sigma  \in \Omega_i(\psi(\mu))$.
Moreover, since $\Omega=\Omega_1 \cup \Omega_2$ and $\sigma \notin \Omega(\psi(\mu'))$,  $\sigma \notin \Omega_i(\psi(\mu'))$.
Since $P_{\Omega}(\mu)=P_{\Omega}(\mu')$, $\Omega \supseteq \Omega_i$, and $\Omega_i$ is delay feasible, by Lemma \ref{Lem1}, $P_{\Omega_i}(\mu)=P_{\Omega_i}(\mu')$.
Therefore, we have $P_{\Omega_i}(\mu)=P_{\Omega_i}(\mu')$, $\sigma  \in \Omega_i(\psi(\mu))$, and $\sigma \notin \Omega_i(\psi(\mu'))$, which violates that $\Omega_i$ is delay feasible.
Similar to the case of (a), we can draw the conclusion that  $\Omega_i$ is not delay feasible, if case (b) holds.
\end{proof}

\subsection{Proof of Proposition \ref{Prop3}}
\begin{proof}
The proof is by induction.
The base case is for $\varepsilon$.
Since $\psi(\varepsilon)=\varepsilon$ and $P_{\Omega}(\varepsilon)=\theta^{\Omega}(\varepsilon)=\varepsilon$, the base case is true.
The induction hypothesis is that for all $\mu \in \mathcal{L}(\tilde{G})$ such that $|\mu|\le n$,  $\psi(\mu)=s \Rightarrow P_{\Omega}(\mu)=\theta^{\Omega}(s).$
We now prove the same is also true for $\mu e \in \mathcal{L}(\tilde{G})$ with $|\mu|=n$.
By the definition of $\tilde{G}$,  $e\in \Sigma$ or $e \in \Sigma_h \cup\Sigma_g$.
We consider each of them separately as follows.

\textsl{Case 1}: $e \in \Sigma$. We have $\psi(\mu e)=\psi(\mu)e=se$.
By definition,  $\theta^{\Omega}(se)=\theta^{\Omega}(s)e$ if $e\in \Omega(s)$, and  $\theta^{\Omega}(se)=\theta^{\Omega}(s)$ if $e \notin \Omega(s)$.
Since $\psi(\mu)=s$, by $\textup{C}1$, it has $e\in\Omega(s) \Leftrightarrow e\in \Omega(g^{-1}(\psi^g(\mu)))$.
Hence, by the definition of $P_{\Omega}(\cdot)$, $P_{\Omega}(\mu e)=P_{\Omega}(\mu)e$ if $e\in \Omega(s)$, and  $P_{\Omega}(\mu e)=P_{\Omega}(\mu)$ if $e \notin \Omega(s)$. Since $P_{\Omega}(\mu)=\theta^{\Omega}(s)$,  $P_{\Omega}(\mu e)=\theta^{\Omega}(s e)$.
\textsl{Case 2}: $e\in \Sigma_h\cup\Sigma_g$. We have $\psi(\mu e)=\psi(\mu)=s$.
By the induction hypothesis, $P_{\Omega}(\mu)=\theta^{\Omega}(s)$.
Since $e\in \Sigma_h\cup\Sigma_g$, by the definition of $P_{\Omega}(\cdot)$, $P_{\Omega}(\mu e)=P_{\Omega}(\mu)$.
Thus,  $P_{\Omega}(\mu e)=\theta^{\Omega}(s)$.
\end{proof}

\subsection{Proof of Proposition \ref{Prop4}}

\begin{proof}
Given any $(q,q') \in \mathcal{T}_{conf,G}(\Theta_{\Omega_1}^{N_o})$,   $\exists s,s'\in \mathcal{L}(G)$ such that $\Theta_{\Omega_1}^{N_o}(s) \cap\Theta_{\Omega_1}^{N_o}(s')\neq \emptyset$, $\delta(q_0,s)=q$, and $\delta(q_0,s')=q'$.
Since $\Theta_{\Omega_1}^{N_o}(s) \cap\Theta_{\Omega_1}^{N_o}(s')\neq \emptyset$, there exist $i,j\le N_o$ such that $\theta^{\Omega_1}(s_{-i})=\theta^{\Omega_1}(s'_{-j})$.
Since $s_{-i},s'_{-j}\in \mathcal{L}(G)$, by Proposition \ref{Prop1}, $\exists \mu,\mu'\in \mathcal{L}(\tilde{G})$ such that $\psi(\mu)=s_{-i}$ and $\psi(\mu')=s'_{-j}$.
Moreover, since $\theta^{\Omega_1}(s_{-i})=\theta^{\Omega_1}(s'_{-j})$  and  $\Omega_1$ is delay feasible, by Proposition \ref{Prop3}, $P_{\Omega_1}(\mu)=P_{\Omega_1}(\mu')$.
Since $\Omega_1 \supseteq \Omega_2$ and $\Omega_2$ is delay feasible, by Lemma \ref{Lem1},  $P_{\Omega_2}(\mu)=P_{\Omega_2}(\mu')$.
Since $\psi(\mu)=s_{-i}$ and $\psi(\mu')=s'_{-j}$, by Proposition \ref{Prop3}, $\theta^{\Omega_2}(s_{-i})=\theta^{\Omega_2}(s'_{-j})$.
Thus, $\Theta_{\Omega_2}^{N_o}(s) \cap \Theta_{\Omega_2}^{N_o}(s')\neq \emptyset$. 
Furthermore, since $\delta(q_0,s)=q$ and $\delta(q_0,s')=q'$, we have $(q,q') \in \mathcal{T}_{conf,G}(\Theta_{\Omega_2}^{N_o})$.
Since $(q,q')$ is arbitrarily given,  $\mathcal{T}_{conf,G}(\Theta_{\Omega_1}^{N_o}) \subseteq \mathcal{T}_{conf,G}(\Theta_{\Omega_2}^{N_o})$.
\end{proof}

\subsection{Proof of Lemma \ref{Lem2}}
\begin{proof}
The proof is by induction.
The base case is true since $P_{\Delta}(\varepsilon)=\theta^{\hat{\Delta}}(\varepsilon)$.
The induction hypothesis is that for all $\mu \in \mathcal{L}(\tilde{G})=\mathcal{L}(W)$ with $|\mu|\le n$, $P_{\Delta}(\mu)=\theta^{\hat{\Delta}}(\mu)$.
We now prove the same is also true for $\mu e$ with $|\mu|=n$.
By definition,  $e \in \Sigma$ or $e \in \Sigma_h \cup \Sigma_f$.

\textsl{Case 1}: $e \in \Sigma$. We write $\delta_W(q_{0,W},\mu)=q_W=(q,\theta_o,\theta_c,x)$.
Then, $x=\delta(q_0,g^{-1}(\psi^g(\mu)))$.
By definition, $\theta^{\hat{\Delta}}(\mu e)=\theta^{\hat{\Delta}}(\mu)e$ if $(q_W,e) \in \hat{\Delta}$, and $\theta^{\hat{\Delta}}(\mu e)=\theta^{\hat{\Delta}}(\mu)$ if $(q_W,e) \notin \hat{\Delta}$.
By the definition of $\hat{\Delta}$,  $(q_W,e) \in \hat{\Delta} \Leftrightarrow (x,e)\in \Delta$.
Thus, $\theta^{\hat{\Delta}}(\mu e)=\theta^{\hat{\Delta}}(\mu)e$ if $(x,e)\in \Delta$ and $\theta^{\hat{\Delta}}(\mu e)=\theta^{\hat{\Delta}}(\mu)$ if $(x,e)\notin \Delta$.
Since  $x=\delta(q_0,g^{-1}(\psi^g(\mu)))$, by the definition of $P_{\Delta}(\cdot)$, $P_{\Delta}(\mu e)=P_{\Delta}(\mu)e$ if $(x,e) \in \Delta$ and $P_{\Delta}(\mu e)=P_{\Delta}(\mu)$ if $(x,e) \notin \Delta$.
Thus, by $P_{\Delta}(\mu)=\theta^{\hat{\Delta}}(\mu)$,  $P_{\Delta}(\mu e)=\theta^{\hat{\Delta}}(\mu e)$.
\textsl{Case 2}: $e \in \Sigma_h \cup \Sigma_f$. $P_{\Delta}(\mu e)=P_{\Delta}(\mu)=\theta^{\hat{\Delta}}(\mu)=\theta^{\hat{\Delta}}(\mu e)$.
\end{proof}

\subsection{Proof of Proposition \ref{Prop6}}
\begin{proof}
Let
$\mathcal{T}=\{(q,q')\in Q^2: \exists  ((q,\theta_o,\theta_c,x),  (q',\theta'_o,\theta'_c,x')) \in \mathcal{T}_{conf,W}(\theta^{\hat{\Delta}})\}.$
We first prove $\mathcal{T}_{conf,G}(P_{\Delta})\subseteq \mathcal{T}$. 
For any  $(q,q')\in \mathcal{T}_{conf,G}(P_{\Delta})$, by (\ref{Eq14}), there exist $\mu,\mu'\in \mathcal{L}(\tilde{G})$ such that $P_{\Delta}(\mu)=P_{\Delta}(\mu')$, $q=\delta(q_0,\psi(\mu))$, and $q'=\delta(q_0,\psi(\mu'))$.
Since $\mathcal{L}(\tilde{G})=\mathcal{L}(W)$, $\mu,\mu'\in \mathcal{L}(W)$.
Let $\delta_{W}(q_{0,W},\mu)=(p,\theta_o,\theta_c,x)$ and  $\delta_{W}(q_{0,W},\mu')=(p',\theta'_o,\theta'_c,x')$.
Since $W=\tilde{G}||G^g$, $\tilde{\delta}(\tilde{q}_0,\mu)=(p,\theta_o,\theta_c)$ and $\tilde{\delta}(\tilde{q}_0,\mu')=(p',\theta'_o,\theta'_c)$.
By Proposition \ref{Prop1}, $p=\delta(q_0,\psi(\mu))=q$ and $p'=\delta(q_0,\psi(\mu'))=q'$.
Meanwhile, since $P_{\Delta}(\mu)=P_{\Delta}(\mu')$, by Lemma \ref{Lem2},  $\theta^{\hat{\Delta}}(\mu)=\theta^{\hat{\Delta}}(\mu')$.
By (\ref{Eq15}), $((p,\theta_o,\theta_c,x), (p',\theta'_o,\theta'_c,x')) \in \mathcal{T}_{conf,W}(\theta^{\hat{\Delta}})$.
Thus, $(p,p')\in \mathcal{T}$.
Since $p=q$ and $p'=q'$, $(q,q')\in \mathcal{T}$.
Since $(q,q')\in \mathcal{T}_{conf,G}(P_{\Delta})$ is arbitrarily given, $\mathcal{T}_{conf,G}(P_{\Delta})\subseteq \mathcal{T}$.

We next prove $\mathcal{T}_{conf,G}(P_{\Delta})\supseteq \mathcal{T}$. 
For any $(q,q')\in \mathcal{T}$, $\exists ((q,\theta_o,\theta_c,x), (q',\theta'_o,\theta'_c,x')) \in \mathcal{T}_{conf,W}(\theta^{\hat{\Delta}})$.
By (\ref{Eq15}), $\exists \mu,\mu' \in \mathcal{L}(W)=\mathcal{L}(\tilde{G})$ such that $\theta^{\hat{\Delta}}(\mu)=\theta^{\hat{\Delta}}(\mu')$, $\delta_W(q_{0,W},\mu)=(q,\theta_o,\theta_c,x)$, and $\delta_W(q_{0,W},\mu')=(q',\theta'_o,\theta'_c,x')$.
Since $W=\tilde{G}||G^g$, $\tilde{\delta}(\tilde{q}_0,\mu)=(q,\theta_o,\theta_c)$ and $\tilde{\delta}(\tilde{q}_0,\mu')=(q',\theta'_o,\theta'_c)$.
By Proposition \ref{Prop1}, $q=\delta(q_0,\psi(\mu))$ and $q'=\delta(q_0,\psi(\mu'))$.
Since $\theta^{\hat{\Delta}}(\mu)=\theta^{\hat{\Delta}}(\mu')$, by Lemma \ref{Lem2}, $P_{\Delta}(\mu)=P_{\Delta}(\mu')$.
By (\ref{Eq14}), $(q,q')\in \mathcal{T}_{conf,G}(P_{\Delta})$.
Since $(q,q')\in \mathcal{T}$ is arbitrarily given, $\mathcal{T}_{conf,G}(P_{\Delta})\supseteq \mathcal{T}$.
\end{proof}

\subsection{Proof of Proposition \ref{Proposition5}}
\begin{proof}
($\Rightarrow$) 
The proof is by contradiction.
If $(\ref{Equation14})$ does not hold, there exists $(q,\theta_o,\theta_c,x)\in Q_W$ with $\tilde{\delta}((q,\theta_o,\theta_c),\sigma)!$ such that (i)  $(q,\sigma) \in \Delta \wedge (x,\sigma) \notin \Delta$, or (ii) $(q,\sigma) \notin \Delta \wedge (x,\sigma) \in \Delta$.
Without loss of generality,  write $\delta_W(q_{0,W},\mu)=(q,\theta_o,\theta_c,x)$.
Since $W=\tilde{G}||G^g$,   $\tilde{\delta}(\tilde{q}_{0},\mu)=(q,\theta_o,\theta_c)$ and $\delta(q_{0},g^{-1}(\psi^g(\mu)))=x$.
If (i), since $\tilde{\delta}((q,\theta_o,\theta_c),\sigma)!$, $\mu\sigma \in \mathcal{L}(\tilde{G})$.
Since $(x,\sigma) \notin {\Delta}$, we have $(\delta(q_0,g^{-1}(\psi^g(\mu))),\sigma) \notin  {\Delta}$.
Moreover, since $\tilde{\delta}(\tilde{q}_{0},\mu)=(q,\theta_o,\theta_c)$, by Proposition \ref{Prop1}, $\delta(q_0,\psi(\mu))=q$.
Since $(q,\sigma)\in \Delta$,  $(\delta(q_0,\psi(\mu)),\sigma) \in \Delta$.
By $\textup{C}3$,  it contradicts that $\Delta$ is a delay feasible SAP.
If (ii), similar to the case of  (i), it contradicts that $\Delta$ is a delay feasible SAP.

If $(\ref{Equation15})$ does not hold, there exists $(q,q') \in \mathcal{T}_{conf,G}(P_{\Delta})$ with $(q,\sigma) \in \Delta$ and $(q',\sigma) \notin \Delta$.
Since $(q,q') \in \mathcal{T}_{conf,G}(P_{\Delta})$, there exist $\mu,\mu'\in \mathcal{L}(\tilde{G})$ such that $\delta(q_0,\psi(\mu))=q$, $\delta(q_0,\psi(\mu'))=q'$, and $P_{\Delta}(\mu)=P_{\Delta}(\mu')$.
Therefore, $\exists \mu,\mu'\in \mathcal{L}(\tilde{G})$ such that $\delta(q_0,\psi(\mu))=q$, $\delta(q_0,\psi(\mu'))=q'$, $P_{\Delta}(\mu)=P_{\Delta}(\mu')$, $(q,\sigma) \in \Delta$, and $(q',\sigma) \notin \Delta$.
By $\textup{C}4$,  it contradicts that $\Delta$ is a delay feasible SAP.

($\Leftarrow$)
Also by contradiction.
Suppose that $\Delta$ is not delay feasible.
Thus, $\textup{C}3$ or $\textup{C}4$ are not true for $\Delta$.
If $\neg \textup{C}3$, there exist $\sigma \in \Sigma$ and $\mu\sigma \in \mathcal{L}(\tilde{G})$ such that (i) $(\delta(q_0,\psi(\mu)),\sigma)\in \Delta \wedge (\delta(q_0,g^{-1}(\psi^g(\mu))),\sigma)\notin \Delta$ or (ii) $(\delta(q_0,\psi(\mu)),\sigma)\notin \Delta \wedge (\delta(q_0,g^{-1}(\psi^g(\mu))),\sigma) \in \Delta$.
Without loss of generality, we write $\tilde{\delta}(\tilde{q}_0,\mu)=(q,\theta_o,\theta_c)$.
By Proposition \ref{Prop1}, $q=\delta(q_0,\psi(\mu))$.
Since $\mu\sigma \in \mathcal{L}(\tilde{G})$, $\tilde{\delta}((q,\theta_o,\theta_c),\sigma)!$.
Since $W=\tilde{G}||G^g$,  ${\delta}_W({q}_{0,W},\mu)=(q,\theta_o,\theta_c,x)$, where $x=\delta(q_0,g^{-1}(\psi^g(\mu)))$.
If (i) is ture, since  $(\delta(q_0,\psi(\mu)),\sigma)\in \Delta$, we have $(q,\sigma)\in \Delta$; Since $(\delta(q_0,g^{-1}(\psi^g(\mu))),\sigma)\notin \Delta$, $(x,\sigma)\notin \Delta$.
That contradicts $(\ref{Equation14})$.
If $\neg \textup{C}4$,  $\exists\sigma \in \Sigma$ and $\mu,\mu' \in \mathcal{L}(\tilde{G})$ such that $P_{\Delta}(\mu)=P_{\Delta}(\mu')$ and (i) $(\delta(q_0,\psi(\mu)),\sigma)\in \Delta \wedge (\delta(q_0,\psi(\mu')),\sigma)\notin \Delta$ or (ii) $(\delta(q_0,\psi(\mu)),\sigma)\notin \Delta \wedge (\delta(q_0,\psi(\mu')),\sigma) \in \Delta$.
By definition, $(\delta(q_0,\psi(\mu)), \delta(q_0,\psi(\mu')))\in \mathcal{T}_{conf, G}(P_{\Delta})$.
That contradicts $(\ref{Equation15})$.
\end{proof}

\subsection{Proof of Proposition \ref{Prop7}}

\begin{proof}
Let ${\Delta}_i$ be the SAP obtained at the end of $i$th-iteration.
We first prove ${\Delta}_i \supseteq\Delta^{\uparrow}$ by induction on $i=0,1,\ldots, n$.
Since $\Delta_0=\Delta \supseteq \Delta^{\uparrow}$, the base case is true.
The induction hypothesis is that for all $\Delta_i$ with $i \le  k$, we have $\Delta_i \supseteq\Delta^{\uparrow}$.
Now, we prove $\Delta_{k+1} \supseteq \Delta^{\uparrow}$ is also true.

For all $(q,\sigma) \in \Delta_k \setminus {\Delta}_{k+1}$ that are removed by Line \ref{lin7} from $\Delta_k$ in the $k+1$st iteration, we have $\textup{D}1 \lor \textup{D}2 \lor \textup{D}3$.
If $\textup{D}1$ is true, there exists $(q,\theta_o,\theta_c,x)\in Q_W$ such that $\tilde{\delta}((q,\theta_o,\theta_c),\sigma)!$ and $(x,\sigma) \notin \Delta_k$.
Since $\Delta_{k} \supseteq \Delta^{\uparrow}$, $(x,\sigma) \notin \Delta^{\uparrow}$.
Without loss of generality,  write $\delta_W(q_{0,W},\mu)=(q,\theta_o,\theta_c,x)$.
Since $W=\tilde{G}||G^g$,  we have $\tilde{\delta}(\tilde{q}_{0},\mu)=(q,\theta_o,\theta_c)$ and $\delta(q_{0},g^{-1}(\psi^g(\mu)))=x$.
Then, we have the following results:
(i) Since $\tilde{\delta}((q,\theta_o,\theta_c),\sigma)!$, $\mu\sigma \in \mathcal{L}(\tilde{G})$;
(ii) Since $(x,\sigma) \notin {\Delta}^{\uparrow}$,  $(\delta(q_0,g^{-1}(\psi^g(\mu))),\sigma) \notin  {\Delta}^{\uparrow}$.
Since $\Omega^{\uparrow}$ is delay feasible, by $\textup{C}3$,  $(\delta(q_0,\psi(\mu)),\sigma) \notin \Delta^{\uparrow}$.
Since $\tilde{\delta}(\tilde{q}_{0},\mu)=(q,\theta_o,\theta_c)$, by Proposition \ref{Prop1}, $\delta(q_0,\psi(\mu))=q$.
Thus,  $(q,\sigma) \notin \Delta^{\uparrow}$.
If $\textup{D}2$ is true, similar to the case of  $\textup{D}1$, $(q,\sigma) \notin \Delta^{\uparrow}$.
If $\textup{D}3$ is true, $\exists (q,q') \in \mathcal{T}_{conf,G}(P_{\Delta_k})$ such that $(q',\sigma)\notin \Delta_k$.
Since $\Delta_k \supseteq {\Delta}^{\uparrow}$,  $(q',\sigma) \notin \Delta^{\uparrow}$.
Since $(q,q') \in \mathcal{T}_{conf,G}(P_{\Delta_k})$, there exist $\mu,\mu'\in \mathcal{L}(\tilde{G})$ such that $\delta(q_0,\psi(\mu))=q$, $\delta(q_0,\psi(\mu'))=q'$, and $P_{\Delta_k}(\mu)=P_{\Delta_k}(\mu')$.
Since $\Delta_k \supseteq {\Delta}^{\uparrow}$ and ${\Delta}^{\uparrow}$ is delay feasible, by Lemma \ref{Lem1}, $P_{\Delta^{\uparrow}}(\mu)=P_{\Delta^{\uparrow}}(\mu')$.
Therefore, $\exists \mu,\mu'\in \mathcal{L}(\tilde{G})$ such that $\delta(q_0,\psi(\mu))=q$, $\delta(q_0,\psi(\mu'))=q'$, $P_{\Delta^{\uparrow}}(\mu)=P_{\Delta^{\uparrow}}(\mu')$, and $(q',\sigma) \notin \Delta^{\uparrow}$.
Since ${\Delta}^{\uparrow}$ is delay feasible, by $\textup{C}4$,  $(\delta(q_0,\psi(\mu)),\sigma)=(q,\sigma) \notin \Delta^{\uparrow}$.

Therefore, for all $(q,\sigma) \in \Delta_k \setminus {\Delta}_{k+1}$, we have $(q,\sigma) \notin \Delta^{\uparrow}$, which implies $\Delta_{k+1} \supseteq \Delta^{\uparrow}$.

Let ${\Delta}_n$ be the SAP obtained in the last iteration of the repeat-until loop on Line 4.
Next, we prove $\Delta_n \subseteq \Delta^{\uparrow}$.
Since $\Delta^{\uparrow}$ is the maximum delay feasible subpolicy of $\Delta$, to prove  $\Delta_n \subseteq \Delta^{\uparrow}$, we only need to prove that $\Delta_n$ is delay feasible.
We prove $\Delta_n$ is delay feasible by contradiction.
Suppose that $\Delta_n$ is not delay feasible.
Thus, $\textup{C}3$ or $\textup{C}4$ are not true for $\Delta_n$.

If $\neg \textup{C}3$, there exist $\sigma \in \Sigma$ and $\mu\sigma \in \mathcal{L}(\tilde{G})$ such that (i) $(\delta(q_0,\psi(\mu)),\sigma)\in \Delta_n \wedge (\delta(q_0,g^{-1}(\psi^g(\mu))),\sigma)\notin \Delta_n$ or (ii) $(\delta(q_0,\psi(\mu)),\sigma)\notin \Delta_n \wedge (\delta(q_0,g^{-1}(\psi^g(\mu))),\sigma) \in \Delta_n$.
Without loss of generality, we write $\tilde{\delta}(\tilde{q}_0,\mu)=(q,\theta_o,\theta_c)$.
By Proposition \ref{Prop1}, $q=\delta(q_0,\psi(\mu))$.
Since $\mu\sigma \in \mathcal{L}(\tilde{G})$, $\tilde{\delta}((q,\theta_o,\theta_c),\sigma)!$.
Since $W=\tilde{G}||G^g$,  ${\delta}_W({q}_{0,W},\mu)=(q,\theta_o,\theta_c,x)$, where $x=\delta(q_0,g^{-1}(\psi^g(\mu)))$.
If (i) is ture, since  $(\delta(q_0,\psi(\mu)),\sigma)\in \Delta_n$,  $(q,\sigma)\in \Delta_n$; Since $(\delta(q_0,g^{-1}(\psi^g(\mu))),\sigma)\notin \Delta_n$, $(x,\sigma)\notin \Delta_n$.
By Line 7, $(q,\sigma)$ will be removed from $\Delta_n$ in the $n+1$st iteration.
If (ii) is true, since  $(\delta(q_0,\psi(\mu)),\sigma)\notin \Delta_n$,  $(q,\sigma)\notin \Delta_n$; Since $(\delta(q_0,g^{-1}(\psi^g(\mu))),\sigma)\in \Delta_n$, $(x,\sigma)\in \Delta_n$.
Also, by Line 7, $(x,\sigma)$ will be removed from $\Delta_n$ in the $n+1$st iteration.

If $\neg \textup{C}4$,  $\exists\sigma \in \Sigma$ and $\mu,\mu' \in \mathcal{L}(\tilde{G})$ such that $P_{\Delta_n}(\mu)=P_{\Delta_n}(\mu')$ and (i) $(\delta(q_0,\psi(\mu)),\sigma)\in \Delta_n \wedge (\delta(q_0,\psi(\mu')),\sigma)\notin \Delta_n$ or (ii) $(\delta(q_0,\psi(\mu)),\sigma)\notin \Delta_n  \wedge (\delta(q_0,\psi(\mu')),\sigma) \in \Delta_n $.
By definition, $(\delta(q_0,\psi(\mu)), \delta(q_0,\psi(\mu')))\in \mathcal{T}_{conf, G}(P_{\Delta_n})$.
By Line 7, if (i) is true, $(\delta(q_0,\psi(\mu)),\sigma)$ will be removed from $\Delta_n$ in the $n+1$st iteration; If (ii) is ture, $(\delta(q_0,\psi(\mu')),\sigma)$  will be removed from $\Delta_n$ in the $n+1$st iteration.

Overall,  $\neg \textup{C}3$ or  $\neg \textup{C}4$ contradict that Algorithm 1 does not terminate until there are no
changes to $\Delta_n$.
Thus, $\Delta_n$ satisfies  $\textup{C}3$ and $\textup{C}4$, which implies $\Delta_n \subseteq \Delta^{\uparrow}$.
\end{proof}

\subsection{Proof of Proposition \ref{Prop8}}

\begin{proof}
    
    Let $\mathcal{T}=\{(x,x')\in Q \times Q: (\exists (q,q') \in \mathcal{T}_{conf,G}(P_{\Delta}))x={R}^{N_o}(q)\wedge x'={R}^{N_o}(q')\}$.
	We first prove $\mathcal{T} \subseteq \mathcal{T}_{conf,G}(\Theta_{\Delta}^{N_o})$.
	For an arbitrary state pair $(x,x')\in \mathcal{T}$,  there exist $(q,q') \in \mathcal{T}_{conf,G}(P_{\Delta})$ and $t,t'\in \Sigma^{\le N_o}$ such that $x=\delta(q,t)$ and $x'=\delta(q',t')$. Since $(q,q')\in \mathcal{T}_{conf,G}(P_{\Delta})$, by (\ref{Eq14}), there exist $\mu, \mu' \in \mathcal{L}(\tilde{G})$ such that $q=\delta(q_0,\psi(\mu)) \wedge q'=\delta(q_0,\psi(\mu')) \wedge P_{\Delta}(\mu)=P_{\Delta}(\mu')$. 
We write $\psi(\mu)=s$ and $\psi(\mu')=s'$.
Then, since $\Delta$ is delay feasible and $P_{\Delta}(\mu)=P_{\Delta}(\mu')$, by Proposition \ref{Prop3}, $\theta^{\Delta}(s)=\theta^{\Delta}(s')$.
Since $q=\delta(q_0,s)$, $q'=\delta(q_0,s')$, $x=\delta(q,t)$, and $x'=\delta(q',t')$, we have $x=\delta(q_0,st)$ and $x'=\delta(q_0,s't')$. 
Moreover, since $t,t' \in \Sigma^{\le N_o}$ and $\theta^{\Delta}(s)=\theta^{\Delta}(s')$, $\Theta_{\Delta}^{N_o}(st)\cap \Theta_{\Delta}^{N_o}(s't')\neq \emptyset$. 
	Therefore,  $\exists st,s't' \in \mathcal{L}(G)$ such that $x=\delta(q_0,st) \wedge x'=\delta(q_0,s't') \wedge  \Theta_{\Delta}^{N_o}(st)\cap \Theta_{\Delta}^{N_o}(s't')\neq \emptyset$, which implies $(x,x')\in \mathcal{T}_{conf,G}(\Theta_{\Delta}^{N_o})$.
	Since $(x,x')\in \mathcal{T}$ is arbitrarily given, $\mathcal{T} \subseteq \mathcal{T}_{conf,G}(\Theta_{\Delta}^{N_o})$.
	
	Next, we prove that $\mathcal{T}_{conf,G}(\Theta_{\Delta}^{N_o}) \subseteq \mathcal{T}$.
    For any $(x,x') \in \mathcal{T}_{conf,G}(\Theta_{\Delta}^{N_o})$, by definition, there exist $s,s' \in \mathcal{L}(G)$ such that $x=\delta(q_0,s) \wedge x'=\delta(q_0,s') \wedge \Theta_{\Delta}^{N_o}(s)\cap \Theta_{\Delta}^{N_o}(s')\neq \emptyset$. 
By the definition of $\Theta_{\Delta}^{N_o}(\cdot)$, $\exists  v,v' \in \Sigma^{\le N_o}$ such that $s=uv \wedge s'=u'v' \wedge \theta^{\Delta}(u)=\theta^{\Delta}(u')$. 
Since $u, u'\in \mathcal{L}(G)$, we write $\delta(q_0,u)=q$ and $\delta(q_0,u')=q'$.
By Proposition \ref{Prop1}, $\exists \mu,\mu'\in \mathcal{L}(\tilde{G})$ such that $\psi(\mu)=u$ and $\psi(\mu')=u'$.
Moreover, since $\theta^{\Delta}(u)=\theta^{\Delta}(u')$ and  $\Delta$ is delay feasible, by Proposition \ref{Prop3}, $P_{\Delta}(\mu)=P_{\Delta}(\mu')$.
By $P_{\Delta}(\mu)=P_{\Delta}(\mu')$, $\delta(q_0,\psi(\mu))=q$, and $\delta(q_0,\psi(\mu'))=q'$, we have $(q,q') \in \mathcal{T}_{conf,G}(P_{\Delta})$. 
Since $\delta(q_0,s)=x$, $\delta(q_0,s')=x$, $s=uv$, and $s'=u'v'$, we have $\delta(q,v)=x$ and $\delta(q',v')=x'$.
    Thus, $\exists (q,q') \in \mathcal{T}_{conf,G}(P_{\Delta})$ such that $\delta(q,v)=x$ and $\delta(q',v')=x'$ with $v,v' \in \Sigma^{\le N_o}$. 
    By the definition of $\mathcal{T}$, $(x,x')\in \mathcal{T}$.
    Since $(x,x')\in \mathcal{T}_{conf,G}(\Theta_{\Delta}^{N_o})$ is arbitrarily given, $\mathcal{T}_{conf,G}(\Theta_{\Delta}^{N_o}) \subseteq \mathcal{T}$.
\end{proof}

\subsection{Proof of Theorem \ref{Theo2}}
\begin{proof}
In Line 1 of Algorithm 2, $\Omega$ is initialized to $Q\times \Sigma_o$.
The set $D$ records the transitions that fail the test for deactivation.
Initially, $D=\emptyset$.
By Line 3 of Algorithm 2, in each iteration, the algorithm takes a transition $(q,\sigma)$ from $\Delta\setminus D$ to potentially remove. 
In Lines 4 and 5 of Algorithm 2, the set $D$ does not change. 
By the if-else loop (Line 6 of Algorithm 2), the transition $(q,\sigma)$ is either removed from $\Delta$ or saved to $D$. 
Since the number of transitions in $\Delta$ is finite, we have $\Delta=D$ in finite steps. 
By Lines 4 and 6 of Algorithm 2, $\Delta^*$ is delay feasible and satisfies $\mathcal{T}_{spec}$. 
Now, we show that the solution is minimal.

Let $\Delta_i$ be the SAP obtained at the end of the $i$th iteration.
When Algorithm 2 terminates,  $\Delta^*=D$.
Give an arbitrary $(q,\sigma)\in \Delta^*=D$. 
When we try to remove $(q,\sigma)$ from $\Delta_i$ at $i+1$st iteration of Algorithm 2,  Line 4 of Algorithm 2 returns the maximum delay feasible subpolicy $\tilde{\Delta}^{\uparrow}$ of $\tilde{\Delta}=\Delta_i \setminus \{(q,\sigma)\}$.
Since $(q,\sigma)\in D$, $\mathcal{T}_{conf,G}(\Theta_{\tilde{\Delta}^{\uparrow}}^{N_o})\cap \mathcal{T}_{spec} \neq \emptyset$.
Sicne $\Delta^* \subseteq {\Delta}_i$, $\Delta^* \setminus \{(q,\sigma)\} \subseteq {\Delta}_i\setminus \{(q,\sigma)\}=\tilde{\Delta}$.
Thus, for any delay feasible $\Delta'\subseteq \Delta^*$ such that $(q,\sigma) \notin \Delta'$,  we have $\Delta' \subseteq \tilde{\Delta}$.
Since $\Delta'$ is delay feasible and $\tilde{\Delta}^{\uparrow}$ is the maximum delay feasible subpolicy of $\tilde{\Delta}$, $\Delta' \subseteq \tilde{\Delta}^{\uparrow}$.
Since $\mathcal{T}_{conf,G}(\Theta_{\tilde{\Delta}^{\uparrow}}^{N_o})\cap \mathcal{T}_{spec} \neq \emptyset$ and $\Delta' \subseteq \tilde{\Delta}^{\uparrow}$, by Proposition \ref{Prop4},  $\mathcal{T}_{conf,G}(\Theta_{\Delta'}^{N_o})\cap \mathcal{T}_{spec} \neq \emptyset$.
Thus, removing $(q,\sigma)$ from $\Delta^*$ will violate the specification $\mathcal{T}_{spec}$. 
Since $(q,\sigma)\in \Delta^*$ is arbitrarily given, $\Delta^*$ is minimal.
\end{proof}

\subsection{Proof of Proposition \ref{Proposition8}}

\begin{proof}
($\Rightarrow$) The proof is by contradiction. 
Suppose that there exists $(q,q')\in \mathcal{T}_{conf,{G}}(\Theta_{\Omega}^{N_o}) \cap \mathcal{T}_{spec} \neq \emptyset$.
Then, there exist $s,s'\in \mathcal{L}(\tilde{G})$ such that $\delta(q_0,s)=q$, $\delta(q_0,s')=q'$, and $\Theta_{\Omega}^{N_{o}}(s)\cap \Theta_{\Omega}^{N_{o}}(s')\neq \emptyset$.
Without loss of generality, we write $q=(x,n_1,\ldots,n_m)$ and $q'=(x',n_1',\ldots,n'_m)$. Since $(q,q')\in \mathcal{T}_{spec}$, there exists $i\in \{1,\ldots,m\}$ such that $n_i=K$ and $n_i'=-1$.
Since $\delta(q_0,s)=q$ and $n_i=K$, by (\ref{Eq1-1}), we can write $s=s_1s_2$ such that $s_1\in \Psi(f_i)$ and $|s_2|\ge K$.
Since $\delta(q_0,s')=q'$ and $n'_i=-1$, by (\ref{Eq1-1}), we have $f_i \notin s'$.
Therefore, $(\exists i \in \mathcal{F})(\exists s_1 \in \Psi(f_i))(\exists s_2 \in \mathcal{L}(G) \setminus s_1)|s_2|\ge K  \wedge [(\exists s' \in \mathcal{L}(G))\Theta_{\Omega}^{N_{o}}(s_1s_2)\cap \Theta_{\Omega}^{N_{o}}(s') \neq \emptyset \wedge f_i \notin s'],$
which contradicts (\ref{Eq13}).

($\Leftarrow$) Also by contradiction. Suppose that (\ref{Eq13}) is not true, i.e., 
$(\exists i \in \mathcal{F})(\exists s_1 \in \Psi(f_i))(\exists s_2 \in \mathcal{L}(G) \setminus s_1)|s_2|\ge K  \wedge [(\exists s' \in \mathcal{L}(G))\Theta_{\Omega}^{N_{o}}(s_1s_2)\cap \Theta_{\Omega}^{N_{o}}(s') \neq \emptyset \wedge f_i \notin s'].$
Since $s_1 \in \Psi(f_i)$, $|s_2|\ge K$, and $f_i \notin s'$, by (\ref{Eq1-1}), there exist $s,s'\in \mathcal{L}(G)$ such that $\delta(q_0,s)=q=(x,n_1,\ldots,n_m)$ with $n_i=K$ and $\delta(q_0,s')=q'=(x',n'_1,\ldots,n'_m)$ with $n_i'=-1$.
Thus, $(q,q')\in \mathcal{T}_{spec}.$
Moreover, since $\Theta_{\Omega}^{N_{o}}(s_1s_2)\cap \Theta_{\Omega}^{N_{o}}(s') \neq \emptyset$, $(q,q')\in \mathcal{T}_{conf,{G}}(\Theta_{\Omega}^{N_o})$.
Therefore,  $(q,q')\in \mathcal{T}_{conf,{G}}(\Theta_{\Omega}^{N_o}) \cap \mathcal{T}_{spec}\neq \emptyset$, which contradicts $\mathcal{T}_{conf,{G}}(\Theta_{\Omega}^{N_o}) \cap \mathcal{T}_{spec} =\emptyset$.
\end{proof}

\subsection{Proof of Proposition \ref{Proposition9}}

\begin{proof}
The proof is by contradiction.
We suppose that  $\bar{\Delta}''$ does not satisfy $\mathcal{T}'_{spec}$, but  $\bar{\Delta}'$ satisfies $\mathcal{T}'_{spec}$.
By (\ref{Eq20}), there exists $(q,q_1,\ldots,q_n) \in \mathcal{T}'_{spec}$ such that $(q,q_1,\ldots,q_n) \in \mathcal{T}_{conf,G}(\Theta^{\bar{\Delta}''})$.
By the definition of $\mathcal{T}_{conf,G}(\Theta^{\bar{\Delta}''})$, there exist $s,s_1,\ldots,s_n$ such that $q=\delta(q_0,s)$, $q_i=\delta(q_0,s_i)$, and $\Theta_{\Delta''_i}^{N_{o,i}}(s)\cap \Theta_{\Delta''_i}^{N_{o,i}}(s_i) \neq \emptyset$ for all $i\in I$.
Thus, $(\delta(q_0,s),\delta(q_0,s_i))\in \mathcal{T}_{conf,G}(\Theta_{\Delta''_i}^{N_{o,i}})$.
Since $\bar{\Delta}' \subseteq \bar{\Delta}''$ and both $\bar{\Delta}'$ and $\bar{\Delta}''$ are delay feasible, by Proposition \ref{Prop4},  $(\delta(q_0,s),\delta(q_0,s_i))\in \mathcal{T}_{conf,G}(\Theta_{\Delta'_i}^{N_{o,i}})$.
Therefore, $\Theta_{\Delta'_i}^{N_{o,i}}(s)\cap \Theta_{\Delta'_i}^{N_{o,i}}(s_i) \neq \emptyset$ for all $i\in I$.
Moreover, since $q=\delta(q_0,s)$ and $q_i=\delta(q_0,s_i)$ for all $i\in I$, $(q,q_1,\ldots,q_n) \in \mathcal{T}_{conf,G}(\Theta^{\bar{\Delta}'})$.
Thus, $(q,q_1,\ldots,q_n) \in \mathcal{T}_{conf,G}(\Theta^{\bar{\Delta}'})\cap \mathcal{T}_{spec}'\neq \emptyset$, which contradicts that $\bar{\Delta}'$ satisfies $\mathcal{T}'_{spec}$.
\end{proof}

\subsection{Proof of Theorem \ref{Theo3}}
\begin{proof}
At each iteration, the algorithm takes an agent $i$ from $I$ and a transition from $\Delta_i$ but not in $D_i$ to remove. 
It is either removed or saved into $D_i$. 
Since $\Delta_i$ is upper bounded by $|TR(G)|$ and $I$ is upper bounded by the number of agents $n$, the algorithm
terminates within finite iterations. 
From Algorithm 3, $\bar{\Delta}^*$ is
dealy feasible and satisfies $\mathcal{T}_{spec}'$. 
Therefore, we only need to show that $\bar{\Delta}^*$ is a minimal SAP.

When Algorithm 3 terminates, we have that $\Delta_i^*=\Delta_i$ for all $i\in I$.
Let $(q,\sigma) \in \Delta_i^*$ for some $i\in I$.
At some time when $\bar{\Delta}\supseteq \bar{\Delta}^*$, we try to remove $(q,\sigma)$ from $\Delta_i$,
but cannot do so. 
By Theorem \ref{Prop2}, there exists a maximum delay SAP $\Delta_i^{\uparrow}\subseteq \Delta_i\setminus \{(q,\sigma)\}$.
Let $\bar{\Delta}^{\uparrow}=[\Delta_1,\ldots,\Delta_{i-1},\Delta_i^{\uparrow},\Delta_{i+1},\ldots,\Delta_n]$.
Then,  $\bar{\Delta}^{\uparrow}$ is delay feasible since $\Delta_i^{\uparrow}$ is delay feasible and, for all $j=1,\ldots,n$, $j\neq i$, 
$\Delta_j$ is either the maximum delay feasible subpolicy of a policy for agent $j$ in a previous iteration or is equal to $\Delta_j^{all}$.
Then, $(q,\sigma)\in D_i$ implies that $\bar{\Delta}^{\uparrow}$ does not satisfy $\mathcal{T}_{spec}'$.
Since $\bar{\Delta}^*\subseteq \bar{\Delta}$, for any delay feasible $\bar{\Delta}'\subseteq \bar{\Delta}^*$ with $(q,\sigma)\notin \bar{\Delta}'$, we have $\bar{\Delta}'\subseteq \bar{\Omega}^{\uparrow}$.
By Proposition \ref{Theo3}, $\bar{\Delta}'$ does not satisfy $\mathcal{T}_{spec}'$.
Thus $(q,\sigma)$ cannot be removed.
Since $i\in I$ and $(q,\sigma)\in \Delta_i^*$ are arbitrary, the proof is completed.
\end{proof}

\end{document}